\documentclass[preprint,10pt]{elsarticle}

\usepackage{amsmath,amsthm}
\usepackage{amsfonts,amssymb}
\usepackage{CJK}
\usepackage{multicol}
\usepackage{multirow}
\usepackage{diagbox}
\usepackage{graphicx}
\usepackage{float}
\usepackage{subfigure}  

\makeatletter
\newenvironment{tablehere}
  {\def\@captype{table}}
 {}

\newenvironment{figurehere}
 {\def\@captype{figure}}
 {}
\makeatother

\usepackage{lineno} 
\allowdisplaybreaks 
\usepackage{epstopdf} 
\usepackage{pgf,tikz}
\usepackage{tikz}
\usetikzlibrary{arrows,shapes,snakes,shadows,positioning,automata,patterns}
\usetikzlibrary{trees,decorations.pathmorphing,decorations.markings}
\usetikzlibrary{backgrounds}

\usepackage[
            pdfstartview=FitH,
            CJKbookmarks=true,
            bookmarksnumbered=true,
            bookmarksopen=true,
            colorlinks, 
            pdfborder=001,   
            linkcolor=blue,
            anchorcolor=blue,
            citecolor=blue
            ]{hyperref}

\newtheorem{theorem}{Theorem}
\newtheorem{assumption}{Assumption}

\newtheorem{lemma}{Lemma}

\newtheorem{definition}{Definition}

\newtheorem{remark}{Remark}

\journal{Automatica}

\begin{document}

\begin{frontmatter}



\title{Distributed Nash Equilibrium Seeking for Noncooperative Games of High-Order Nonlinear Multi-Agent Systems Over Weight-Unbalanced Digraphs
\tnoteref{label1}}
\tnotetext[label1]{This work was supported by the National Natural Science Foundation of China under Grant 61803385 and the Hunan Provincial Natural Science Foundation of China under Grant 2019JJ50754.}


\author[1]{Zhenhua Deng\fnref{label2}}
\ead{zhdeng@amss.ac.cn}
\author[1]{Jin Luo}
\ead{luoboof@csu.edu.cn}
\fntext[label2]{Corresponding author, Tel. +8613657357739}

\address[1]{School of Automation, Central South University, Changsha, 410075, China}

\begin{abstract}
In this paper, we investigate the noncooperative games of multi-agent systems.
Different from  existing noncooperative games, our formulation involves the high-order nonlinear dynamics of players, and the communication topologies among players are weight-unbalanced digraphs.
Due to the high-order nonlinear dynamics and the weight-unbalanced digraphs, existing Nash equilibrium seeking algorithms cannot solve our problem. In order to seek the Nash equilibrium of the noncooperative games, we propose two distributed algorithms based on state feedback and output feedback, respectively. Moreover, we analyze the convergence of the two algorithms with the help of variational analysis and Lyapunov stability theory. By the two algorithms, the high-order nonlinear players exponentially converge to the Nash equilibrium.
Finally, two simulation examples illustrate the effectiveness of the algorithms.
\end{abstract}


\end{frontmatter}


\begin{multicols}{2}
\section{Introduction}
%
%
%
%
Noncooperative games arise widely in varieties of fields, such as economic markets, smart grids, communication networks, and social networks (see
 \cite{Jensen2010AggregativeGA,Gharesifard2016PriceBasedCA,Deng2021DistributedSM,Barrera2015DynamicIF,Lou2016NashEC,Ghaderi2014OpinionDI}). 
In noncooperative games,
each player aims to selfishly minimize its own cost function, depending on its own decisions and the decisions of other players
(see  \cite{Frihauf2012NashES,Facchinei2010GeneralizedNE,Li2020DistributedGN,Koshal2016DistributedAF}).
Consequently, Nash equilibrium seeking is a critical problem in noncooperative games, which has received more and more attention recently
(see \cite{Deng2019DistributedAF,ZENG2019GNE,Gadjov2019APA,Lu2019DistributedAF}).

In the past few years, numbers of Nash equilibrium seeking algorithms have been proposed for noncooperative games.
For example,
\cite{Ye2017DistributedNE} presented consensus-based  seeking algorithms for unconstrained noncooperative games.
\cite{Gadjov2019APA} designed projection-based  seeking algorithms for noncooperative games with convex set constraints.
\cite{ZENG2019GNE} proposed primal-dual based  seeking algorithms for noncooperative games with inequality constraints.
\cite{Lu2019DistributedAF} developed gradient-based  seeking algorithms for noncooperative games with differentiable cost functions.
\cite{Deng2021nonsmooth} exploited subgradient-based  seeking algorithms for noncooperative games with nondifferentiable cost functions.
\cite{Zhang2020DistributedNE} proposed internal model-based  seeking algorithms for disturbed noncooperative games.

Communication topologies among players have a significant influence on the design and analysis of distributed algorithms.
Most of existing  distributed Nash equilibrium seeking algorithms for noncooperative games rely on undirected graphs or weight-balanced digraphs, such as  \cite{Frihauf2012NashES, Li2020DistributedGN,Koshal2016DistributedAF,Deng2019DistributedAF,ZENG2019GNE,Gadjov2019APA,Lu2019DistributedAF, Ye2017DistributedNE,Deng2021nonsmooth,Zhang2020DistributedNE}.
However, an algorithm may be ineffective over weight-unbalanced digraphs, even though it is convergent over undirected graphs or weight-balanced digraphs (see \cite{Zhu2019ContinuousTimeCA}).
Moreover, it is well-known that weight-unbalanced digraphs are the extensions of undirected graphs and weight-balanced digraphs, which have wider applications and are easier to implement in engineering practices.
Consequently,  it is necessary to investigate the noncooperative games over weight-unbalanced digraphs.

As an integration of computation, communication, and physical processes, cyber-physical systems (CPSs) occur commonly in many engineering practices,
such as transportation systems, power systems, and mobile sensor networks (see \cite{Kim2012CyberPhysicalSA,Zhang2018DCR,
Zhang2015DistributedPA}).
In the background of CPSs,  physical systems are used to autonomously perform distributed tasks, and hence
increasing attention has been paid to distributed algorithms combined with the dynamics of physical systems.
For instance,
\cite{Wang2021DOC,Zhang2017DistributedOC} developed distributed optimization algorithms for Euler-Lagrange systems.
\cite{Wang2020DistributedPA,Deng2021DAD} proposed distributed resource allocation algorithms for second-order and high-order linear systems, respectively.
\cite{Deng2021DistributedEA,Romano2020DynamicNS} designed distributed Nash equilibrium seeking algorithms for second-order nonlinear systems  and high-order linear systems, respectively.
On the other hand,  high-order nonlinear systems can characterize many real physical systems, such as jerk systems \cite{Eichhorn1998TransformationsON}, synchronous generators \cite{2002Nonlinear}, vehicles \cite{2006Vehicle},
and turbine generators \cite{Fang2011BacksteppingbasedNA}, which extend first-order and second-order nonlinear systems, and high-order linear systems.
Therefore, high-order nonlinear systems have been universally studied.
For example,
\cite{Li2017OFD,Rezaee2019StationaryCC} investigated containment problems and consensus problems of high-order nonlinear systems, respectively.
However, to the best of our knowledge,
there are no results about noncooperative games of high-order nonlinear multi-agent systems over weight-unbalanced digraphs.
Furthermore, without furthering involving the control of high-order nonlinear dynamics and the weight-unbalanced digraphs,
existing Nash equilibrium seeking algorithms for noncooperative games, including  \cite{Deng2021DistributedSM,Frihauf2012NashES,Li2020DistributedGN,Koshal2016DistributedAF,Deng2019DistributedAF,ZENG2019GNE,Gadjov2019APA,Lu2019DistributedAF,Ye2017DistributedNE,Deng2021nonsmooth,Zhang2020DistributedNE,Deng2021DistributedEA,Romano2020DynamicNS},   can not solve the problem.

The objective of this paper is to investigate the noncooperative games of high-order nonlinear multi-agent systems over weight-unbalanced digraphs, and design distributed algorithms to seek the Nash equilibrium of the noncooperative games.
The main contributions of this paper are summarized as follows.

\begin{itemize}
	\item [(i)] This paper studies a noncooperative game with high-order nonlinear players over weight-unbalanced digraphs.
	The formulation extends the well-known noncooperative games, such as
	\cite{Deng2019DistributedAF,Zhang2020DistributedNE,Deng2021DistributedEA,Romano2020DynamicNS}, by adding the high-order nonlinear dynamics of players.
	Moreover, in our problem, the communication networks among players are weight-unbalanced digraphs, which extend undirected graphs and weight-balanced digraphs required by many results, such as \cite{Frihauf2012NashES,Li2020DistributedGN,Koshal2016DistributedAF,Deng2019DistributedAF,ZENG2019GNE,Gadjov2019APA,Lu2019DistributedAF,Ye2017DistributedNE,Deng2021nonsmooth,Zhang2020DistributedNE}.
	Owing to the high-order nonlinear dynamics and the weight-unbalanced digraphs, existing Nash equilibrium seeking algorithms, such as \cite{Deng2021DistributedSM,Lou2016NashEC,Frihauf2012NashES,Li2020DistributedGN,Koshal2016DistributedAF,Deng2019DistributedAF,ZENG2019GNE,Gadjov2019APA,Lu2019DistributedAF,Ye2017DistributedNE,Deng2021nonsmooth,Zhang2020DistributedNE,Deng2021DistributedEA,Romano2020DynamicNS}, cannot be applied to this problem.
	
	\item [(ii) ]The distributed algorithm design and analysis of our problem face the following challenges.
	(i) The dynamics of players are high-order and nonlinear, and the cost functions of players are also nonlinear. Hence, the whole closed-loop system is a complicated high-order nonlinear system;
	(ii) The convergence analysis of existing Nash equilibrium seeking algorithms over undirected graphs or weight-balanced digraphs depend on the conditions that $1_N^TL=0_N$, $L^T=L$, and the positive semi-definition of $L$ or $L+L^T$, while the Laplacian matrix $L$ of weight-unbalanced digraphs is nonsymmetric and nonpositive definite.
	\item [(iii)] This paper proposes a distributed state-based algorithm and a distributed output-based algorithm for the high-order nonlinear players to seek the Nash equilibrium of noncooperative games.
    Moreover, this paper analyzes the convergence of the two algorithms.
	Under the two algorithms, the high-order nonlinear players exponentially converge to the Nash equilibrium, even the high-order nonlinear dynamics contain uncertain parameters.
	Furthermore, the distributed output-based algorithm only depends on the output variables of the high-order nonlinear players rather than all  state infromation.
\end{itemize}


This paper is organized as follows.
Section \ref{sec.basis} introduces preliminaries, and formulates the problem.
Section \ref{sec.r} presents two distributed Nash equilibrium seeking algorithms, and analyzes their convergence.
Section \ref{sec.s} provides two examples to verify the algorithms.
Finally, Section \ref{sec.c} summarizes the conclusion.

\section{Preliminaries and Formulation}\label{sec.basis}

This section first presents some preliminaries, and then formulates the studied problem.
\subsection{Preliminaries}

\textsl{Notations.}
$\mathbb{R}$ and  $\mathbb{R}^{n}$ denote the set of real numbers and the $n$-dimensional Euclidean space, respectively.
$\mathbb{Z}_{+}$ is the set of nonnegative integers.
$\otimes$ represents the Kronecker product.
$\Vert A  \Vert $  and $\Vert x  \Vert $  are the spectral norm of matrix A and the standard Euclidean norm of vector $x$, respectively.
$x^T$ is the transpose of $x$.
$col(x_1,\dots,x_n)=[x_1^T,\dots,x_n^T]^T$ with $x_i \in \mathbb{R}^{n}$.
${I_n}$ and $0_{n \times n}$ denote the identity  and  zero matrices with $n \times n$ dimensions, respectively.
${0_n}$ and ${1_n}$ denote the column vectors of $n$ ones and zeros, respectively.
$diag\{v_1,v_2,\dots,v_N\}$ denotes a diagonal matrix with the elements $v_1,v_2,\dots,v_N$ being all on its main diagonal.

\subsubsection{Graph Theory (see  \cite{GODSIL2004Algebraic})}

 A directed graph (or simple a digraph) of $N$ nodes is denoted by $\mathcal{G} :=\left\lbrace {\mathcal{V},\mathcal{E}},\mathcal{A}\right\rbrace$, where $\mathcal{V} =\left\{1,\dots,N\right\}$, $\mathcal{E}\subseteq\mathcal{V}\times\mathcal{V}$ are the node and edge sets, respectively, and $\mathcal{A}$ is the adjacency matrix.
$(i,j)\in\mathcal{E}$ is an edge of $\mathcal{G}$ if node $i$ can receive information from node $j$.
Denote adjacency matrix $\mathcal{A} := \left\lbrack a_{ij} \right\rbrack_{N\times N}$, where $a_{ij}> 0$ if $(i,j) \in \mathcal{E}$, and  $a_{ij} = 0$, otherwise.
Here $a_{ii}=0$ for any $i\in\mathcal{V}$, which indicates no self-connection in the graph.
Besides, for an edge $(i,j)\in\mathcal{E}$, $i$ is called the out-neighbor of $j$, and $j$ is called the in-neighbor of $i$.
The weighted in-degree and weighted out-degree of node $i$ are $d_{in}^i=\sum_{j=1}^Na_{ij}$ and  $d_{out}^i=\sum_{j=1}^Na_{ji}$, respectively.
Note that $d_{in}^i$ may not equal to $d_{out}^i$ for weight-unbalanced digraphs.
A directed path from $v_0$ and $v_l$ is defined as a sequence of nodes $v_0,\dots,v_l$, such that $(v_k,v_{k+1})\in\mathcal{E}$, $\forall k\in\{0,1,\dots,l-1\}$.
Particularly, a digraph $\mathcal{G}$ is said to be strongly connected if there exists a directed path between any two distinct nodes.
The Laplacian matrix of $\mathcal{G}$ is $L:=\mathcal{D}_{in} - \mathcal{A}$, where $\mathcal{D}_{in} := diag\{d_{in}^1,\dots,d_{in}^N\}\in \mathbb{R}^{N\times N}$. Obviously, $L1_N=0_N$.

For a weighted-unbalanced digraph $\mathcal{G}$, we have the following results.
\begin{lemma}\label{lem.M}
	(see \cite{Hu2007LeaderfollowingCO}) Suppose the weight-unbalanced digraph $\mathcal{G}$ is strongly connected with the Laplacian matrix $L$, and
	$M$ is a diagonal matrix with nonnegative diagonal elements and at least one diagonal element being positive.
	Then, the following conditions hold.
	\begin{enumerate}
		\item[(i)] $L+M$ is positive definite;
		\item[(ii)] There exists a positive definite matrix $Q\in\mathbb{R}^{N\times N}$ such that $Q(L+M) + (L+M)^TQ = I_N$.
	\end{enumerate}
\end{lemma}

 \subsubsection{Variational Analysis (see \cite{2004Variational})}

A function $f: {\mathbb{R}^m} \to \mathbb{R}$ is convex if
\begin{align*}
	f(\alpha x + (1-\alpha)y) \leq \alpha f(x) + (1-\alpha)f(y)\\
	\forall x, y \in {\mathbb{R}^m}, \forall \alpha \in [0,~1].
\end{align*}

A map $F:{\mathbb{R}^m} \to \mathbb{R}^m$ is  $\omega$-strongly monotone ($\omega>0$) if
\begin{equation*}
	{(x - y)^T}(\nabla F(x) - \nabla F(y)) \ge \omega{\left\| {x - y} \right\|^2}, ~ \forall \, x, y \in \mathbb{R}^m.
\end{equation*}

A map $F:{\mathbb{R}^m} \to \mathbb{R}^m$ is $\theta$-Lipschitz ($\theta>0$) if
\begin{equation*}
	\left\| {F(x) - F(y)} \right\| \le \theta{\left\| {x - y} \right\|},~ \forall \, x, y \in \mathbb{R}^m.
\end{equation*}

For a continuously differentiable map $F:{\mathbb{R}^m} \to \mathbb{R}^m$, it is (strongly) strictly monotone if and only if its Jacobian matrix $\mathcal{J}F$ is (uniformly) positive definite (see \cite{2004Variational}).

The solution $x\in\mathbb{R}^m$ of the variational inequality $VI(\mathbb{R}^m,F)$, $F:{\mathbb{R}^m} \to \mathbb{R}^m$, satisfies the following condition  (see \cite{2004Variational}).
\begin{equation}\label{VI}
	(y-x)F(x)\geq 0, \forall y\in\mathbb{R}^m.
\end{equation}
In addtion, $x$ satisfies the condition \eqref{VI} if and only if $x$ is the solution of $F(x)=0_m$

The following lemma is about the uniqueness of the solution of $VI(\mathbb{R}^m,F)$.
\begin{lemma}
	(see \cite{2003Finite}) If a continuous map $F:{\mathbb{R}^m} \to \mathbb{R}^m$ is $\omega$-strongly monotone,  then $VI(\mathbb{R}^m,F)$ has a unique solution.
\end{lemma}

\subsection{Problem Formulation}
Consider a noncooperative game of $N$ players over a weight-unbalanced digraph $\mathcal{G}$.
Player $i\in\mathcal{V}$ has a differentiable cost function $J_i(x_i,x_{-i}):\mathbb{R}^{Nm}\to \mathbb{R}$, where $x_i\in\mathbb{R}^m$ is the decision of player $i$, $x_{-i}=col(x_1,\dots,x_{i-1},x_{i+1},\dots,x_N)$.
 The purpose of player $i$ is to minimize its own cost function $J_i(x_i,x_{-i})$ by changing its own decision $x_i$.  Specifically, player $i$ faces the following optimization problem.
\begin{equation} \label{q}
	\min \limits_{{x_i}\in\mathbb{R}^{m}}\mathop{J_i(x_i,x_{-i})}.
\end{equation}

The Nash equilibrium of the noncooperative game \eqref{q} is defined as follows (referring to \cite{Facchinei2010GeneralizedNE, Deng2019DistributedAF}).
\begin{definition}\label{def.NE}
	A decision profile ${x}^*:=(x_i^*,x_{-i}^*)$ is a Nash equilibrium of the noncooperative game \eqref{q} if
	\begin{equation*}
		J_i(x_i^*,x_{-i}^*) \leq J_i(x_i,x_{-i}^*), \forall x_i\in\mathbb{R}^m, i \in \mathcal{V}.
	\end{equation*}
\end{definition}

The above definition indicates that, at a Nash equilibrium, no player can decrease its cost function by changing its own decision unilaterally.

Player $i\in\mathcal{V}$ has the following uncertain $n$th-order $(n\geq 1)$ nonlinear dynamics.
\begin{equation}\label{sys}
	\begin{aligned}
		\dot x_i =& {x_i}^{(1)}\\
		\vdots & \\
		\dot{x}_i^{(n-1)} =& {x_i}^{(n)}\\
		{x_i}^{(n)} =&f_i(x_i,{x_i}^{(1)},\dots,{x_i}^{(n-1)},w_i)+u_i
	\end{aligned}
\end{equation}
where $x_i^{(l)},l\in\{1,\dots,n\}$, is the $l$th order derivative of $x_i$;
$f_i(\cdot,\dots,\cdot,w_i): \mathbb{R}^m\times\dots\times\mathbb{R}^m\to \mathbb{R}^m$ is the nonlinear function; $w_i\in\mathbb{R}^{n_{w_i}}$ is the unknown constant parameter with  $n_{w_i}\in\mathbb{Z}_{+}$;
and $u_i\in\mathbb{R}^m$ is the control input of player $i$.

\begin{remark}
	In engineering, the dynamics of many physical systems, such as jerk systems \cite{Eichhorn1998TransformationsON}, synchronous generators \cite{2002Nonlinear}, vehicles \cite{2006Vehicle}, and turbine generators \cite{Fang2011BacksteppingbasedNA}, can be described by \eqref{sys}.	
	Besides, the first-order and second-order nonlinear systems, and high-order linear systems investigated in \cite{Zhang2020DistributedNE,Deng2021DistributedEA,Romano2020DynamicNS} are the spacial cases of the high-order nonlinear system \eqref{sys}.
	Therefore, our results can be straightforwardly applied to these systems.
\end{remark}

The aim of this paper is to design distributed algorithms such that the output of high-order nonlinear player \eqref{sys} converges to the Nash equilibrium of the noncooperative game \eqref{q}.

\begin{remark}
	In contrast to well-studied noncooperative games, such as \cite{Lou2016NashEC,Frihauf2012NashES,Gadjov2019APA,Ye2017DistributedNE},
    this formulation involve the high-order nonlinear dynamics of players, which is in keeping with the development of CPSs.
	Moreover, the communication topologies among players in our formulation are weight-unbalanced digraphs, which are weaker than undirected graphs and weight-balanced digraphs required by \cite{Frihauf2012NashES, Li2020DistributedGN,Koshal2016DistributedAF,Deng2019DistributedAF,ZENG2019GNE,Gadjov2019APA,Lu2019DistributedAF, Ye2017DistributedNE,Deng2021nonsmooth,Zhang2020DistributedNE}.
Without considering the control of high-order dynamics and the weight-unbalanced digraphs, 	existing Nash equilibrium seeking algorithms are ineffective for our problem, such as \cite{Deng2021DistributedSM,Lou2016NashEC,Frihauf2012NashES,Li2020DistributedGN,Koshal2016DistributedAF,Deng2019DistributedAF,ZENG2019GNE,Gadjov2019APA,Lu2019DistributedAF,Ye2017DistributedNE,Deng2021nonsmooth,Zhang2020DistributedNE,Deng2021DistributedEA,Romano2020DynamicNS}.
Because of the  high-order nonlinear dynamics, the nonlinear cost functions, and the weighted-unbalance digraphs, it is not easy to design distributed Nash equilibrium seeking algorithms for our problem and analyze their convergence.
\end{remark}

Some assumptions about the communication topologies, the high-order nonlinear dynamics, and the cost functions are presented as follows.
\begin{assumption}\label{ass.G}
	The weight-unbalanced digraph $\mathcal{G}$ is strongly connected.
\end{assumption}

\begin{assumption}\label{ass.non}
	$f_i(\cdot)$ satisfies the following Lipschitz-like condition.
	\begin{align*}	
		&\|f_i(x_i,{x_i}^{(1)},\dots,{x_i}^{(n-1)},w_i)-f_i(\check{x}_i,{\check{x}_i}^{(1)},\dots,{\check{x}_i}^{(n-1)},w_i)\|\\
		&\leq l_x\|x_i-\check{x}_i\|+\sum_{l=1}^{n-1} l_{x^{(l)}}\|x_i^{(l)}-\check{x}_i^{(l)}\|
	\end{align*}
	where $l_x,l_{x^{(l)}} \geq0$.
\end{assumption}
\begin{remark}
	Assumption \ref{ass.non} generalizes the Lipschitz condition satisfied by numerous well-known systems, such as
	Chen systems, Chua's circuits, Lorenz systems, pendulum systems and car-like robots (see \cite{Ghapani19} and references therein). What is more,  if the partial derivatives of $f_i(\cdot)$ are uniformly bounded, such as (piecewise-) linear continuous functions, then Assumption \ref{ass.non} holds.
\end{remark}

\begin{assumption}\label{ass.fun}
	The cost function $J_i(x_i,x_{-i})$ is continuously differentiable in $x$ and convex in $x_i$ for every fixed $x_{-i}$.
\end{assumption}
\begin{assumption}\label{ass.func}
	The pseudo-gradient $\mathcal{F}({x})$ is $\omega$-strongly monotone and $\theta$-Lipschitz continuous,
	where
	\begin{equation}\label{F(x)}
		\mathcal{F}({x})=col(\nabla_{x_1}J_1(x_1,x_{-1}),\dots,\nabla_{x_N}J_N(x_N,x_{-N})).
	\end{equation}
\end{assumption}

\begin{remark}
	Assumptions \ref{ass.fun} and \ref{ass.func} indicate the existence and uniqueness of the Nash equilibrium of the noncooperative game \eqref{q}, which were widely used in noncooperative games (see \cite{Deng2019DistributedAF,Ye2017DistributedNE,Zhang2020DistributedNE,Deng2021DistributedEA}).
	Also, Assumptions \ref{ass.fun} and \ref{ass.func} hold in variety of practical engineering problems, such as Nash-Cournot games of generation systems (see \cite{Deng2019DistributedAF}), energy consumption games in smart grids (see \cite{Deng2021DistributedSM}), and rate allocation games in communications (see \cite{Yin2011NashEP}).
\end{remark}

With Assumptions \ref{ass.fun} and \ref{ass.func}, we have the following results about the Nash equilibrium of the noncooperative game \eqref{q}.
\begin{lemma}\label{lm.VE}
	(see \cite{2010Generalized}) Under Assumption \ref{ass.fun}, the solution of $VI(\mathbb{R}^{Nm},\mathcal{F})$ coincides with the Nash equilibrium of the noncooperative game \eqref{q}.
\end{lemma}
\begin{lemma}\label{lm.NE}
	Under Assumption \ref{ass.fun}, ${x}^*=col(x_1^*,\dots,x_N^*)$ is the Nash equilibrium of the noncooperative game \eqref{q} if and only if
	\begin{equation}\label{NE}
		\nabla_{x_i} J_i(x_i^*,x_{-i}^*)=0_m, \forall i\in\mathcal{V} \quad or\quad  \mathcal{F}({x}^*)=0_{Nm}.
	\end{equation}
\end{lemma}
\begin{proof}
	Based on Lemma \ref{lm.VE}, the solution of $VI(\mathbb{R}^{Nm},\mathcal{F})$ is equivalent to the Nash equilibrium of the noncooperative game \eqref{q}, and the converse is true. Furthermore, the solution of $VI(\mathbb{R}^{Nm},\mathcal{F})$ satisfies $\mathcal{F}({x})=0_{Nm}$. Consequently, we obtain the conclusion.
\end{proof}

\section{Main Results}\label{sec.r}
This section proposes two distributed Nash equilibrium seeking algorithms for the noncooperative game \eqref{q} with the high-order nonlinear player \eqref{sys},
and then analyzes their convergence.

\subsection{Distributed State-Based Algorithm}\label{al.sb}

Before giving our algorithms, the following characteristic polynomial associated with real coefficients $(k_1,\ldots,k_{n-1})$ is defined such that its roots are in the open left half plane (LHP).
\begin{align}\label{cp}
	p(s) :=s^{n-1} +k_{n-1} s^{n-2} +\ldots +k_{2} s +k_{1}.
\end{align}

The above definition indicates that the following companion matrix $A$ is Hurwitz.
\begin{align}\label{eqA}
	A=\left[
	\begin{array}{ccc}
		0_{n-2} & \vline & I_{n-2}  \\
		\hline
		-k_1 & \vline &
		\left[
		\begin{array}{ccc}
			-k_2 &  \ldots & -k_{n-1}
		\end{array}
		\right]
	\end{array}
	\right].
\end{align}

The following  well-known lemma about the companion matrix $A$ is used later. 
\begin{lemma}\label{lem.cp}
	(see \cite{1998Linear}) There exists a positive definite symmetric matrix $P:=[p_{ij}]_{(n-1) \times (n-1)}$ such that $PA +A^T P =-I_{n-1}$ is satisfied.
\end{lemma}

When the players' state variables are obtainable, the distributed Nash equilibrium seeking algorithm for player $i\in\mathcal{V}$ is designed as follows.
\begin{subequations}\label{s.al}
	\begin{align}
		u_i =& -\sum_{l=1}^{n-1} \varepsilon^{n-l} k_l x_i^{(l)} - \alpha_1 \nabla_{x_i}J_i(x_i,\hat{x}_{-i}) - \alpha_2 y_i \label{s.alga}\\
		\dot {y}_i =& \sum_{l=1}^{n-1} \varepsilon^{1-l} k_l x_i^{(l)} + \frac{\alpha_1}{\varepsilon^{n-1}} \nabla_{x_i}J_i(x_i,\hat{x}_{-i}) \label{s.algb}\\
		\dot{\hat x}_{j}^i =& -\alpha_3(\sum_{k=1}^{N}a_{ik}({\hat x}_{j}^i-{\hat x}_{j}^k) + a_{ij}({\hat x}_{j}^i-x_j)) \label{s.algc}
	\end{align}
\end{subequations}
where ${\hat x}_{-i} = col({\hat x}_1^i,\dots,{\hat x}_{i-1}^i,{\hat x}_{i+1}^i,\dots,{\hat x}_{N}^i)$ with ${\hat x}_{j}^i$ being the estimate of player $i$ on $x_j$ of player $j$; $\alpha_1, \alpha_2, \alpha_3$ and $\varepsilon$ are parameters (to be determined later); $k_1,\ldots, k_{n-1}$ are  the coefficients of the characteristic polynomial $p(s)$ defined in \eqref{cp} with roots in the open LHP.

Substituting \eqref{s.al} into \eqref{sys}, we have the following closed-loop system.
\begin{subequations}\label{csys}
	\begin{align}
		\dot x =& {x}^{(1)}\\
		\vdots & \nonumber\\
		{x}^{(n)} =& f(x,\dots,x^{(n-1)},w)-\sum_{l=1}^{n-1} \varepsilon^{n-l} k_l x^{(l)}\\
		& - \alpha_1 {F}(x,\hat{ x})- \alpha_2 y \\
		\dot y =& \sum_{l=1}^{n-1} \varepsilon^{1-l} k_l x^{(l)} + \frac{\alpha_1}{\varepsilon^{n-1}} {F}(x,\hat{ x}) \\
		\dot{\hat{ x}} =& -\alpha_3(((L\otimes I_N)\otimes I_m)\hat{ x} + (M\otimes I_m)(\hat{ x}-1_N\otimes x))
	\end{align}
\end{subequations}
where
	$x=col(x_1,\ldots,x_N)$;
	$x^{(l)}= col(x_1^{(l)},\ldots, x_N^{(l)})$ with $l\in\{1,\dots,n\}$;
	$y= col(y_1,\ldots, y_N)$;
	$\hat{ x}=col(\hat{ x}_1,\ldots,\hat{ x}_N)$;
	$\hat{ x}_i=col(\hat{x}_{1}^i,\ldots, \hat{x}_{N}^i)$;
	$f(x,\dots,x^{(n-1)},w)=col(f_1(x_1,\dots,x_1^{(n-1)},w_1),\dots,f_N(x_N,\dots,x_N^{(n-1)},w_N))$;
	${F}(x,\hat{ x})=col(\nabla_{x_1}J_1(x_1,{\hat x}_{-1}),\dots,\nabla_{x_N}J_N(x_N,{\hat x}_{-N}))$; and $M=diag\{M_1,\dots,M_N\}$ with $M_i=diag\{a_{i1},a_{i2},\dots,a_{iN}\}$.
In what follows, we use $f$ to denote $f(x,\dots,x^{(n-1)},w)$, if there is no ambiguity.

For the system \eqref{csys}, we have the following theorem.
\begin{theorem}\label{tm.equ1}
	Under Assumptions \ref{ass.G} and \ref{ass.fun},
	if $({x}^*,x^{*(1)},\dots,x^{*(n-1)},y^*,{\hat{ x}}^*)$  is an equilibrium point of \eqref{csys},
	then ${x}^*$ is a Nash equilibrium of the noncooperative game \eqref{q}. Conversely, if ${x}^*$ is a Nash equilibrium of the noncooperative game \eqref{q}, there exists $(x^{*(1)},\dots,x^{*(n-1)},y^*,\hat{ x}^*)\in\mathbb{R}^{Nm}\times\dots\times\mathbb{R}^{Nm}\times\mathbb{R}^{Nm}\times\mathbb{R}^{N^2m}$ such that $({x}^*,x^{*(1)},\dots,x^{*(n-1)},y^*,{\hat{ x}}^*)$ is an equilibrium point of \eqref{csys}.
\end{theorem}
\begin{proof}
	(i) The equilibrium point of \eqref{csys} satisfies the following conditions.
	
	\begin{subequations}
		\label{equi}	
		\begin{align}
			0_{Nm}	&= {x^*}^{(l)}, l\in\{1,\dots,n-1\}  \label{equia}\\
			0_{Nm} &= f^*-\sum_{l=1}^{n-1} \varepsilon^{n-l} k_l x^{*(l)} - \alpha_1 {F}(x^*,{\hat{ x}}^*)- \alpha_2 y^*   \label{equib}\\
			0_{Nm} &= \sum_{l=1}^{n-1} \varepsilon^{1-l} k_l x^{*(l)} + \frac{\alpha_1}{\varepsilon^{n-1}} {F}(x^*,{\hat{ x}}^*)   \label{equic}\\
			\nonumber 0_{N^2 m} &= -\alpha_3(((L\otimes I_N)\otimes I_m){\hat{ x}}^*\\
			&+ (M\otimes I_m)({\hat{ x}}^*-1_N\otimes x))   \label{equid}
		\end{align}
	\end{subequations}
	where $f^*=f(x^*,\ldots,x^{*(n-1)},w)$.

	Since the weight-unbalanced digraph $\mathcal{G}$ is strong connected, and $M$ is a diagonal matrix with at least one diagonal element being positive, it results from Lemma \ref{lem.M} that $(L\otimes I_N)\otimes I_m + (M\otimes I_m)$ is positive definite.
	Then, \eqref{equid} indicates that $((L\otimes I_N)\otimes I_m + (M\otimes I_m)){\hat{ x}}^* =  ((L\otimes I_N)\otimes I_m + (M\otimes I_m))(1_N\otimes x^*)$.
	Therefore, we obtain ${\hat{x}}^*=1_N\otimes x^*$, i.e., ${F}(x^*,{\hat{ x}}^*)={F}(x^*,1_N\otimes x^*)$.
	Further,  based on \eqref{equia}, \eqref{equic}, and  the definition of ${F}(x,\hat{ x})$, we have ${F}(x^*,{\hat{ x}}^*)=\mathcal{F}({ x}^*)=0_{Nm}$.
	Subsequently, according to Lemma \ref{lm.NE}, ${x}^*$ is a Nash equilibrium of the noncooperative game \eqref{q}.
	
	(ii) Conversely, when $x^*$ is a Nash equilibrium of the noncooperative game \eqref{q}, we have $\mathcal{F}({x}^*)=0_{Nm}$.
	Take ${\hat{ x}}^*=1_N\otimes x^*$ and $x^{*(l)}=0_m$, and thus \eqref{equia}, \eqref{equic} and \eqref{equid} hold.
By choosing $y^*=\frac{1}{\alpha_2}f^*$, \eqref{equib} is satisfied.
	Consequently, $({x}^*,x^{*(1)},\dots,x^{*(n-1)},y^*,{\hat{ x}}^*)$ is an equilibrium point of \eqref{csys}.
\end{proof}

By Theorem \ref{tm.equ1}, the following result is yielded.
\begin{lemma}\label{csys.c}
 The high-order nonlinear player \eqref{sys} converges to the Nash equilibrium of the noncooperative game \eqref{q} under the algorithm \eqref{s.al}, if the system \eqref{csys} is stabilized to its equilibrium point.
\end{lemma}
\begin{proof}
	The high-order nonlinear player \eqref{sys} and its control input \eqref{s.al} constitute the closed-loop system \eqref{csys}.
    Furthermore, based on Theorem \ref{tm.equ1}, the  Nash equilibrium of the noncooperative game \eqref{q} coincide with the equilibrium point of \eqref{csys}.
    Hence, we obtain the result.
\end{proof}

According to Lemma \ref{csys.c}, the convergence of the high-order nonlinear player \eqref{sys} to the Nash equilibrium of the noncooperative game \eqref{q} can be proved by analyzing the stability of the system \eqref{csys}.

Before analyzing the stability of \eqref{csys}, the following lemma about ${F}(x,\hat{ x})$ is given.
\begin{lemma}\label{lem.func}
	Under Assumption \ref{ass.func}, ${F}(x,\hat{ x})$ is $\theta$-Lipschitz in $\hat x$.
\end{lemma}
\begin{proof}
	Define $\hat{y}_{-i}$, $\hat{y}_{j}^i$ and $\hat{ y}$ in the same way as $\hat{x}_{-i}$, $\hat{x}_{j}^i$ and $\hat{ x}$.
    By Assumption \ref{ass.func}, we can obtain
	$\|\nabla_{x_i}J_i(x_i,{\hat x}_{-i})-\nabla_{x_i}J_i(x_i,{\hat y}_{-i})\|^2 \leq \|{F}(\check { x}_i)-{F}(\check{ y}_i)\|^2 \leq \theta^2\|\check { x}_i-\check { y}_i\|^2=\\\theta^2\|\hat{ x}_{-i}-\hat{ y}_{-i}\|^2\leq\theta^2\|\hat{ x}_{i}-\hat{ y}_{i}\|^2, \forall i\in\{1,\dots,N\}$, where $\check { x}_i=col(\hat{x}_1^i,\ldots,x_i,\ldots,\hat{x}_N^i)$ and $\check{ y}_i=col(\hat{y}_1^i,\ldots,x_i,\ldots,\hat{y}_N^i)$. Therefore, it is not hard to obtain ${\|{F}(x,\hat{ x}) -{F}(x,\hat{ y})\|}^2\\=\ \  \sum_{i=1}^N\ {\|\  \nabla_{x_i}J_i(x_i,{\hat x}_{-i})\ -\ \nabla_{x_i}J_i(x_i,{\hat y}_{-i})\ \|}^2 \leq \\\sum_{i=1}^N\theta^2{\|\hat{ x}_i-\hat{ y}_i\|}^2 =\theta^2{\|\hat{ x}-\hat{ y}\|}^2$,
	which implies that ${F}(x,\hat{ x})$ is $\theta$-Lipschitz in $\hat x$.
\end{proof}

By virtue of Lemma \ref{lem.func}, the following result is obtained.
\begin{theorem}\label{lem.con1}
	Under Assumptions  \ref{ass.G}, \ref{ass.non}, \ref{ass.fun} and \ref{ass.func}, the high-order nonlinear player \eqref{sys} with the algorithm \eqref{s.al} exponentially converges to the Nash equilibrium of the noncooperative game \eqref{q}.
\end{theorem}
\begin{proof}
	Without loss of generality, let $m=1$ for simplicity. Let
	$\bar{x} = x-x^*$, 	
	$\bar{y} = y -y^*$, 	
	$\bar{x}^{(l)} = {x}^{(l)}-{x}^{*(l)}$, 	
	$\bar{\hat { x}} = \hat{ x} - I_N\otimes x$,
	$\bar{f}(\bar{x},\ldots,\bar{x}^{(n-1)},w) = f(x,\ldots,x^{(n-1)},w) - f(x^*,\ldots,x^{*(n-1)},w)$,
	where $l \in \{1, \ldots, n-1\}$.
For simplicity, in what follows, we use $\bar f$ to represent $\bar{f}(\bar{x},\ldots,\bar{x}^{(n-1)},w)$.
	
	Combining \eqref{csys} with \eqref{equi} , we have
	\begin{subequations}\label{tsys}
		\begin{align}
			\dot {\bar x} =& {\bar{x}}^{(1)}\\
			\vdots & \nonumber\\
			{\bar{x}}^{(n)} =& \bar{f}-\sum_{l=1}^{n-1} \varepsilon^{n-l} k_l \bar{x}^{(l)} - \alpha_1 h- \alpha_2 \bar{y}\\
			\dot {\bar y} =&  \sum_{l=1}^{n-1} \varepsilon^{1-l} k_l \bar{x}^{(l)} + \frac{\alpha_1}{\varepsilon^{n-1}}h\\
			\dot {\bar{\hat{ x}}} =& -\alpha_3((L\otimes I_N) + M)\bar{\hat{ x}} -1_N\otimes \dot{\bar{x}}
		\end{align}
	\end{subequations}
	where
	$ h= {F}(x,\hat{ x}) -{F}(x^*,\hat{ x}^*)$.
	
	After the above transformation, the equilibrium point of \eqref{tsys} is the origin.
	
	Let
	\begin{align*}
		\tilde{x} =&\  col(\tilde{x}^{(1)}, \ldots, \tilde{x}^{(n-1)}) \\
		\tilde{x}^{(l)} =& \ col(\tilde{x}_1^{(l)}, \ldots, \tilde{x}_N^{(l)})
	\end{align*}
	where
	$\tilde{x}_i^{(l)} = \frac{1}{\varepsilon^{l}} {\bar x}_i^{(l)}$ with
	$l \in \{1, \ldots, n-1\}$.
	
	Thus \eqref{tsys} can be rewritten as
	\begin{subequations}\label{gsys}
		\begin{align}
			\dot {\bar x} =& {\varepsilon} \tilde{x}^{(1)}\\
			\dot {\tilde{x}} =& \varepsilon (A \otimes I_{N}) \tilde{x} + \frac{1}{\varepsilon^{n-1}} (b \otimes I_{N}) ( \bar f -\alpha_1 h - \alpha_2 \bar{y}) \\
			\dot {\bar y} =& \sum_{l=1}^{n-1} \varepsilon k_l \tilde{x}^{(l)} + \frac{\alpha_1}{\varepsilon^{n-1}}h \\
			\dot {\bar{\hat{ x}}} =& -\alpha_3((L\otimes I_N) + M)\bar{\hat{ x}} -1_N\otimes \dot{\bar{x}}
		\end{align}
	\end{subequations}
	where $b=\begin{bmatrix}
		0_{n-2}^T & 1
	\end{bmatrix}^T$ and $A$ is defined in \eqref{eqA}. Clearly, it is equivalent between \eqref{tsys} and \eqref{gsys}.
	
	Take the following Lyapunov function for \eqref{gsys}.
	\begin{align*}
		V_1
		=& \tilde{x}^T (P \otimes I_{N}) \tilde{x}
		+ {\bar{\hat{ x}}}^TQ{\bar{\hat{ x}}}
		+\frac{1}{2}\|{\tilde{x}}^{(n-1)} +\bar y\|^2 \\
		&+\frac{1}{2}\| k_1 {\bar x} +\sum_{i=1}^{n-2}  k_{i+1} {\tilde{x}} ^{(i)} + {\tilde{x}}^{(n-1)}\|^2
	\end{align*}
	where $P$ and $Q$ are positive definite matrices such that $PA +A^T P =-I_{n-1}$ (see Lemmas \ref{lem.cp}) and $Q(L
	\otimes I_N +M) + (L\otimes I_N+M)^TQ = I_{N^2}$ (see Lemma \ref{lem.M}), respectively.
	
	The derivative of $V_1$ along \eqref{gsys} is
	\begin{align*}
		{\dot V}_1 =\nonumber
		& -\alpha_3 \|\bar{\hat{ x}}\|^2
		-2\bar{\hat{ x}}^T(1_N\otimes \dot{\bar x})
		-\varepsilon \|\tilde{x}\|^2
		-\frac{\alpha_2}{\varepsilon^{n-1}} \|\bar y\|^2 \\
		\nonumber
		&+\frac{1}{\varepsilon^{n-1}} {{\bar y}^T}{\bar f}
		+\frac{k_1}{\varepsilon^{n-1}} {{\bar x}^T}({\bar f}-\alpha_1h-\alpha_2\bar y)	\\
		& +\frac{1}{\varepsilon^{n-1}} \bigg(\sum_{l=1}^{n-2} (2 p_{l(n-1)} +k_{l+1}) \tilde{x}^{(l)}\nonumber \\
		& +(2 p_{(n-1)(n-1)} +2) \tilde{x}^{(n-1)}\bigg)^T (\bar f -\alpha_2 \bar y) \nonumber \\
		& -\frac{\alpha_1}{\varepsilon^{n-1}} \bigg(\sum_{l=1}^{n-2} (2 p_{l(n-1)} +k_{l+1}) \tilde{x}^{(l)}\nonumber \\
		&  +(2 p_{(n-1)(n-1)} +1) \tilde{x}^{(n-1)}\bigg)^T h.
	\end{align*}
	
	Owing to $\mathcal{F}(x)={F}(x,1_N\otimes x)$ and $\mathcal{F}({ x}^*)={F}(x^*,{\hat x}^*)$, based on the $\omega$-strongly monotonicity of $\mathcal{F}(x)$ (see Assumption \ref{ass.func}), and the $\theta$-Lipschitz continuity of $F(x,\hat x)$ (see Lemma \ref{lem.func}), we obtain
	\begin{align}\label{eq1.xh}
		\nonumber	
		&-\frac{k_1\alpha_1}{\varepsilon^{n-1}}\bar{{ x}}^Th \\
		\nonumber	
		=& -\frac{k_1\alpha_1}{\varepsilon^{n-1}}\bar{{ x}}^T({F}(x,\hat{ x}) - {F}(x,1_N\otimes x) + \mathcal{F}(x)- \mathcal{F}({ x}^*))\\
		\leq&
		-\frac{k_1\omega\alpha_1}{2\varepsilon^{n-1}}\|\bar x\|^2 + \frac{k_1\theta^2\alpha_1}{2\omega\varepsilon^{n-1}}\|\bar{\hat{ x}}\|^2.
	\end{align}
	By Assumption \ref{ass.non}, we have
	\begin{subequations}\label{eq1.ass2}
		\begin{align}\label{eq1.yf}
			\nonumber
			\frac{1}{\varepsilon^{n-1}}{\bar y}^T\bar f 	\nonumber
			\leq&
			\frac{1}{\varepsilon^{n-1}}\|\bar y\|(l_x\|\bar x\|+\sum_{l=1}^{n-1}l_{x^{(l)}}\|{\bar x}^{(l)}\|) \\
			\leq&
			\frac{1}{2\varepsilon^{n-1}} \bigg((l_x+\sum_{l=1}^{n-1}l_{x^{(l)}}\varepsilon^{2l+1-n})\|\bar y\|^2 \nonumber \\
			&+l_x\|\bar x\|^2 + \sum_{l=1}^{n-1}l_{x^{(l)}}\varepsilon^{n-1}\|{\tilde x}^{(l)}\|^2\bigg)
		\end{align}
		\begin{align}\label{eq1.xf}
			\frac{k_1}{\varepsilon^{n-1}}{\bar x}^T\bar f \nonumber 	
			\leq&
			\frac{k_1}{\varepsilon^{n-1}}\|\bar x\|\bigg(l_x\|\bar x\|+\sum_{l=1}^{n-1}l_{x^{(l)}}\|{\bar x}^{(l)}\|\bigg) \nonumber \\
			\leq&
			\frac{k_1}{2\varepsilon^{n-1}}\bigg( (2l_x+\sum_{l=1}^{n-1}l_{x^{(l)}}\varepsilon^{2l+1-n})\|\bar x\|^2 \nonumber \\
			&+\sum_{l=1}^{n-1}l_{x^{(l)}}\varepsilon^{n-1}\|{\tilde x}^{(l)}\|^2\bigg).
		\end{align}
	\end{subequations}
	By Young's inequality, we have
	\begin{subequations}\label{eq1.young}
		\begin{align}
& -2\bar{\hat{ x}}^T(1_N\otimes \dot{\bar {x}})
			\leq \varepsilon^2\|\bar{\hat{ x}}\|^2 +N\|{\tilde x}^{(1)}\|^2 \label{eq1.hatxx} \\
& -\frac{k_1\alpha_2}{\varepsilon^{n-1}}{\bar x}^T\bar y
			\leq
			\frac{{k_1}^2\alpha_2}{2\varepsilon^{n-1}} \|\bar x\|^2
			+ \frac{\alpha_2}{2\varepsilon^{n-1}} \|\bar y\|^2 \label{eq1.xy} \\
& \bigg(\sum_{l=1}^{n-2} (2 p_{l(n-1)} +k_{l+1}) {\tilde x}^{(l)} \nonumber \\
			& ~~  +(2 p_{(n-1)(n-1)} +2) {\tilde x}^{(n-1)} \bigg)^T \bar f \nonumber \\
		  & \leq
			 \sum_{l=1}^{n-2} |2 p_{l(n-1)} +k_{l+1}|\|{\tilde x}^{(l)}\|\bigg(l_x\|\bar x\|+\sum_{l=1}^{n-1}l_{x^{(l)}}\|{\bar x}^{(l)}\|\bigg) \nonumber\\
			&\quad  +|2 p_{(n-1)(n-1)} +2| \|{\tilde x}^{(n-1)}\|\bigg(l_x\|\bar x\|+\sum_{l=1}^{n-1}l_{x^{(l)}}\|{\bar x}^{(l)}\|\bigg) \nonumber\\
			 & \leq
		\frac{1}{2}\bigg({l_x}+\sum_{l=1}^{n-1}l_{x^{(l)}}\varepsilon^{2l+1-n}\bigg)\bigg(\sum_{l=1}^{n-2} (2 p_{l(n-1)} +k_{l+1})^2 \|{\tilde x}^{(l)}\|^2 \nonumber\\
			&\quad  +(2 p_{(n-1)(n-1)} +2)^2 \|{\tilde x}^{(n-1)}\|^2 \bigg) \nonumber\\
			&\quad +\frac{n-1}{2}\bigg(l_x\|\bar x\|^2 + \sum_{l=1}^{n-1}l_{x^{(l)}}\varepsilon^{n-1} \|{\tilde x}^{(l)}\|^2\bigg) \label{eq1.pxf} \\
			&\bigg(\sum_{l=1}^{n-2} (2 p_{l(n-1)} +k_{l+1}) {\tilde x}^{(l)} \nonumber \\
			& ~~  +(2 p_{(n-1)(n-1)} +1) {\tilde x}^{(n-1)} \bigg)^T h \nonumber \\
			& \leq
			\bigg(\sum_{l=1}^{n-2} |2 p_{l(n-1)} +k_{l+1}| \|{\tilde x}^{(l)}\| \nonumber\\
			& \quad  +|2 p_{(n-1)(n-1)} +1| \|{\tilde x}^{(n-1)}\| \bigg)^T\|\bar{\hat{ x}}\| \nonumber \\
			 & \leq
			\frac{\alpha_1\theta}{2}\bigg(\sum_{l=1}^{n-2} (2 p_{l(n-1)} +k_{l+1})^2 \|{\tilde x}^{(l)}\|^2\nonumber\\
			& \quad +(2 p_{(n-1)(n-1)} +1)^2 \|{\tilde x}^{(n-1)}\|^2) + ({n-1})\|\bar{\hat{ x}}\|^2\bigg) \label{eq1.pxh} \\
			&-\alpha_2 \bigg(\sum_{l=1}^{n-2} (2 p_{l(n-1)} +k_{l+1}) {\tilde x}^{(l)} \nonumber \\
			& \qquad \quad  +(2 p_{(n-1)(n-1)} +2) {\tilde x}^{(n-1)} \bigg)^T\bar y \nonumber \\
			& \leq
		    \alpha_2\bigg(\sum_{l=1}^{n-2} |2 p_{l(n-1)} +k_{l+1}| \|{\tilde x}^{(l)}\| \nonumber\\
			& \quad +|2 p_{(n-1)(n-1)} +2| \|{\tilde x}^{(n-1)}\| \bigg)^T\|\bar y\| \nonumber \\
			& \leq
			\frac{n\alpha_2}{2}\bigg(\sum_{l=1}^{n-2} (2 p_{l(n-1)} +k_{l+1})^2 \|{\tilde x}^{(l)}\|^2\nonumber\\
			&\quad +(2 p_{(n-1)(n-1)} +2)^2 \|{\tilde x}^{(n-1)}\|^2 \bigg )+ \frac{(n-1)\alpha_2}{2n}\|\bar y\|^2. \label{eq1.pxy}
		\end{align}
	\end{subequations}	

	Combining \eqref{eq1.xh}-\eqref{eq1.young}, we have 
	\begin{align*}
		\dot{V}_1 \leq -\rho_1\|\bar x\|^2 -\rho_2\|{\tilde x}\|^2 -\rho_3\|\bar y\|^2 -\rho_4\|\bar{\hat{ x}}\|^2
	\end{align*}
	where
	\begin{align*}
		\rho_1&=\frac{1}{2\varepsilon^{n-1}}(k_1\omega\alpha_1-{k_1}^2\alpha_2-l_x(2k_1+n)\\&\qquad\qquad-\sum_{l=1}^{n-1}k_1l_{x^{(l)}}\varepsilon^{2l+1-n})\\
		\rho_2&=\frac{1}{2\varepsilon^{n-1}}(2{\varepsilon^{n}}-2N\varepsilon^{n-1}-l_{x}^{max}\varepsilon^{n-1}(n+k_1) \\&\qquad\qquad-{\bar p}(l_x+\sum_{l=1}^{n-1}l_{x^{(l)}}\varepsilon^{2l+1-n}+\theta\alpha_1+n\alpha_2))\\
		\rho_3&=\frac{1}{2\varepsilon^{n-1}}(\frac{\alpha_2}{n}-l_x-\sum_{l=1}^{n-1}l_{x^{(l)}}\varepsilon^{2l+1-n})\\
		\rho_4&=\alpha_3
		-\varepsilon^2-\frac{\theta\alpha_1}{2\varepsilon^{n-1}}(n-1+\frac{k_1\theta}{\omega})
	\end{align*}
$p_{i(n-1)}$ is the last element of the $i$th row vector of matrix $P$;
		$\bar p=  \max \{ (2p_{1(n-1)}+k_{2})^2, \ldots,(2p_{(n-2)(n-1)}+k_{n-1})^2, (2p_{(n-1)(n-1)}+2)^2 \big\}$;
		and ${l}_x^{max} = max\{l_x,l_{x^{(1)}},\dots,l_{x^{(n-1)}}\}$.

	Obviously, by choosing suitable sufficiently large positive constants $\varepsilon$, $\alpha_1$, $\alpha_2$ such that $\varepsilon^{n-1}<\alpha_2<\alpha_1<\varepsilon^n$, we have $\rho_1$, $\rho_2$, $\rho_3>0$.
	Further, $\rho_4>0$ holds by taking an appropriate $\alpha_3$.
	Hence, there exist positive constants $\varepsilon$, $\alpha_1$, $\alpha_2$ and $\alpha_3$ such that $\dot{V}_1\leq0$.
Moreover, because $V_1$ and $\dot{V}_1$ are quadratic form, the high-order nonlinear player \eqref{sys} with the algorithm \eqref{s.al} globally exponentially converges to the Nash equilibrium of the noncooperative game \eqref{q}.
\end{proof}

\subsection{Distributed Output-Based Algorithm}\label{al.ob}
When only the players' output variables are obtainable, the distributed Nash equilibrium seeking algorithm for player $i\in\mathcal{V}$ is designed as follows.
\begin{subequations}\label{o.al}
	\begin{align}
		u_i =& -\sum_{l=1}^{n-1} \varepsilon^{n-l} k_l z_i^{(l)} - \alpha_1 \nabla_{x_i}J_i(x_i,\hat{x}_{-i}) - \alpha_2 y_i \label{s2.alga}\\
		\dot {y}_i =& \sum_{l=1}^{n-1} \varepsilon^{1-l} k_l z_i^{(l)} + \frac{\alpha_1}{\varepsilon^{n-1}} \nabla_{x_i}J_i(x_i,\hat{x}_{-i}) \label{s2.algb}\\
		\dot{z}_i = & z_i^{(1)}+\frac{\varepsilon\beta_1}{\mu}(x_i-z_i)\\
		\dot{z}_i^{(1)}=& {z}_i^{(2)} + \frac{\varepsilon^2\beta_{2}}{\mu^2}(x_i-z_i)\\
		\vdots & \nonumber\\
		{z}_i^{(n)}=&  \frac{\varepsilon^n\beta_{n}}{\mu^{n}}(x_i-z_i)\\
		\dot{\hat x}_{j}^i =& -\alpha_3(\sum_{k=1}^{N}a_{ik}({\hat x}_{j}^i-{\hat x}_{j}^k) + a_{ij}({\hat x}_{j}^i-x_j)) \label{s2.algc}
	\end{align}
\end{subequations}
where $k_1,\ldots, k_{n-1}$ are  the coefficients of the characteristic polynomial $p(s)$ defined in \eqref{cp} with roots in the open LHP;
 $\beta_{1},\dots,\beta_{n}$ are the coefficients of $s^{n} +\beta_{1} s^{n-1} +\ldots +\beta_{n-1} s +\beta_{n}=0$ with roots in the open LHP; $\alpha_1$, $\alpha_2$, $\alpha_3$, $\varepsilon$ and $\mu$ are parameters (to be determined later).

The algorithm \eqref{o.al} is similar to the algorithm \eqref{s.al}, except from $z_i$ and ${z}_i^{(l)}$. In the algorithm \eqref{o.al}, $z_i$ and ${z}_i^{(l)}$ are used to estimate the state variables.

Combining \eqref{sys} with \eqref{o.al}, we obtain the following closed-loop system.
\begin{subequations}\label{csys1}
	\begin{align}
		\dot x =& {x}^{(1)}\\
		\vdots & \nonumber\\
		{x}^{(n)} =& f-\sum_{l=1}^{n-1} \varepsilon^{n-l} k_l z^{(l)} - \alpha_1 {F}(x,\hat{x})- \alpha_2 y \\
		\dot y =& \sum_{l=1}^{n-1} \varepsilon^{1-l} k_l z^{(l)} + \frac{\alpha_1}{\varepsilon^{n-1}} {F}(x,\hat{ x}) \\
		\dot{z} = & z^{(1)}+\frac{\varepsilon\beta_1}{\mu}(x-z)\\
		\dot{z}^{(1)}=& {z}^{(2)} + \frac{\varepsilon^2\beta_{2}}{\mu^2}(x-z)\\
		\vdots & \nonumber\\
		{z}^{(n)}=&  \frac{\varepsilon^n\beta_{n}}{\mu^{n}}(x-z)\\
		\dot{\hat{x}} =& -\alpha_3(((L\otimes I_N)\otimes I_m)\hat{x} + (M\otimes I_m)(\hat{ x}-1_N\otimes x))
	\end{align}
\end{subequations}
where $z=col(z_1,\dots,z_N)$ and $z^{(l)}=col(z_1^{(l)},\dots,z_N^{(l)})$ with $l\in\{1,\dots,n\}$.

There are following results about the system \eqref{csys1}.
\begin{theorem}\label{tm.eq2}
	Under Assumptions \ref{ass.G} and \ref{ass.fun}, if $(x^*,x^{*(1)},\dots,x^{*(n-1)},z^*,z^{*(1)},\dots,z^{*(n-1)},y^*,{\hat{x}}^*)$ is an equilibrium point of \eqref{csys1}, $x^*$ is a Nash equilibrium of the noncooperative game \eqref{q}. Conversely, if $x^*$ is a Nash equilibrium of the nonoperative game \eqref{q}, there exists $(x^{*(1)},\dots,x^{*(n-1)},z^*,z^{*(1)},\dots,z^{*(n-1)},y^*,{\hat{x}}^*)\in \mathbb{R}^{Nm} \times \dots\times\mathbb{R}^{Nm} \times \mathbb{R}^{Nm} \times \mathbb{R}^{Nm} \times \dots\times\mathbb{R}^{Nm} \times\mathbb{R}^{Nm} \times \mathbb{R}^{N^2m}$ such that $(x^*,x^{*(1)},\dots,x^{*(n-1)},z^*,z^{*(1)},\dots,z^{*(n-1)},y^*,{\hat{x}}^*)$ is an equilibrium point of \eqref{csys1}.
\end{theorem}
\begin{proof}
At the equilibrium point of \eqref{csys1}, we have $x^*=z^*$ and $x^{*(l)}=z^{*(l)}=0_{Nm}$, $l\in\{1,\dots,n-1\}$.
 Therefore, similarly to the proof of Theorem \ref{tm.equ1}, the conclusion is yielded.
\end{proof}

\begin{lemma}\label{lm.ec}
	The high-order nonlinear player \eqref{sys} converges to the Nash equilibrium of the nonoperative game \eqref{q} under the algorithm \eqref{o.al}, if
	the system \eqref{csys1} is stabilized to its equilibrium point.
\end{lemma}
\begin{proof}
	Similarly to the proof of Lemma \ref{csys.c}, the conclusion can be directly deduced.
\end{proof}

With Lemma \ref{lm.ec}, we  can obtain the following result.
\begin{theorem}\label{Th.con2}
	Under Assumptions \ref{ass.G},  \ref{ass.non}, \ref{ass.fun} and \ref{ass.func}, the high-order nonlinear player \eqref{sys} with the algorithm \eqref{o.al} exponentially converges to the Nash equilibrium of the noncooperative game \eqref{q}.
\end{theorem}
\begin{proof}
	Without loss of generality, for simplicity, let $m=1$. Transforming the equilibrium point of \eqref{csys1} into the origin, we have
	\begin{subequations}\label{tsys1}
		\begin{align}
			\dot {\bar x} =& {\bar{x}}^{(1)}\\
			\vdots & \nonumber\\
			{\bar{x}}^{(n)} =& \bar{f}-\sum_{l=1}^{n-1} \frac{\varepsilon^{n-l}k_l}{\mu^{l-n+1}}  \tilde{z}^{(l)} -\sum_{l=1}^{n-1} \varepsilon^{n-l} k_l \bar{x}^{(l)}-\alpha_1 h- \alpha_2 \bar{y}\\
			\dot {\bar y} =&  \sum_{l=1}^{n-1} \frac{\varepsilon^{1-l}k_l}{\mu^{l-n+1}}  \tilde{z}^{(l)} + \sum_{l=1}^{n-1} \varepsilon^{1-l} k_l \bar{x}^{(l)} + \frac{\alpha_1}{\varepsilon^{n-1}}h\\
			\dot{\tilde z} = & -\frac{\varepsilon\beta_1}{\mu}\tilde z + \frac{1}{\mu}\tilde z^{(1)}\\
			\dot{\tilde z}^{(1)}=&  - \frac{\varepsilon^2\beta_{2}}{\mu}\tilde z+ \frac{1}{\mu}\tilde z^{(2)}\\
			\vdots & \nonumber\\
			\tilde{z}^{(n)}=&  -\frac{\varepsilon^n\beta_{n}}{\mu}\tilde z + \sum_{l=1}^{n-1} \frac{\varepsilon^{n-l}k_l}{\mu^{l-n+1}}  \tilde{z}^{(l)} + \sum_{l=1}^{n-1} \varepsilon^{n-l} k_l \bar{x}^{(l)}\nonumber \\&- \bar{f} +\alpha_1 h + \alpha_2 \bar{y}\\
			\dot {\bar{\hat{ x}}} =& -\alpha_3((L\otimes I_N) + M)\bar{\hat{ x}} -1_N\otimes \dot{\bar{x}}
		\end{align}
	\end{subequations}
	where $\tilde z = \frac{1}{\mu^{n-1}}(z-x)$ and $\tilde z ^{(l)}=\frac{1}{\mu^{n-(l+1)}}(z^{(l)}-x^{(l)})$ with $l\in\{1,\dots,n-1\}$.
	
	Let
	\begin{align*}
		\Delta\tilde{x} =\ & col(\tilde{z},\Delta\tilde{x}^{(1)}, \ldots, \Delta\tilde{x}^{(n-1)}) \\
		\Delta\tilde{x}^{(l)} =\  & col(\Delta\tilde{x}^{(l)}_1, \ldots, \Delta\tilde{x}^{(l)}_N)
	\end{align*}
	where $\Delta\tilde{x}^{(l)}_i = \frac{1}{\varepsilon^{l}}\tilde{z}_i^{(l)}$ with $l\in\{1,\dots,n-1\}$.
	
	Then, \eqref{tsys1} can be rewritten as
	\begin{subequations}\label{ct.sys}
		\begin{align}
			\dot {\bar x} =& {\varepsilon} \tilde{x}^{(1)}\\
			\dot {\tilde{x}} =& \varepsilon (A \otimes I_{N}) \tilde{x} + \frac{1}{\varepsilon^{n-1}} (b \otimes I_{N}) ( \bar f -\alpha_1 h - \alpha_2 \bar{y}\nonumber\\
			& -\sum_{l=1}^{n-1}\frac{\varepsilon^nk_l}{\mu^{l+1-n}} \Delta\tilde{x}^{l}) \\
			\dot {\bar y} =& \sum_{l=1}^{n-1}  \varepsilon k_l \tilde{x}^{(l)} + \sum_{l=1}^{n-1}\frac{\varepsilon^nk_l}{\mu^{l+1-n}}\Delta\tilde{x}^{l}+  \frac{\alpha_1}{\varepsilon^{n-1}}h \\
			\Delta\dot{\tilde{x}} =& \frac{\varepsilon}{\mu}(B\otimes I_{N})\Delta\tilde{x} - \frac{1}{\varepsilon^{n-1}}(b_n\otimes I_N)(\bar f -\alpha_1 h - \alpha_2 \bar{y}\nonumber\\
			& -\sum_{l=1}^{n-1}\frac{\varepsilon^nk_l}{\mu^{l+1-n}}\Delta\tilde{x}^{l}-\sum_{l=1}^{n-1} \varepsilon^nk_l \tilde{x}^{(l)} )	 \\
			\dot {\bar{\hat{ x}}} =& -\alpha_3((L\otimes I_N) + M)\bar{\hat{ x}} -1_N\otimes \dot{\bar{x}}
		\end{align}
	\end{subequations}
	where 
	$\bar\mu=max\{0, k_1\mu^{n-2},\dots,k_{n-2}\mu,k_{n-1}\}$; $b_n=\begin{bmatrix}
		0_{n-1}^T & 1
	\end{bmatrix}^T$ and
	\begin{align*}
		B=\left[
		\begin{array}{ccc}
			\left[
			\begin{array}{ccc}
				-{\beta}_1 &  \ldots & -{\beta}_{n-1}
			\end{array}
			\right]^T  & \vline & I_{n-1}  \\
			\hline
			-{\beta}_n &\vline & 0_{n-1}	
		\end{array}
		\right].
	\end{align*}
	It is obviously that $B$ is Hurwitz, and hence there exists a positive definite matrix $R$ such that $RB+B^TR=-I_n$.
%

	Take the following Lyapunov function for \eqref{ct.sys}.
	\begin{align*}
		V_2=V_1+\Delta\tilde{x}^T(R\otimes I_N)\Delta\tilde{x}.
	\end{align*}
	
	The derivative of $V_2$ along \eqref{ct.sys} is
	\begin{align*}
		{\dot V}_2 =\nonumber
		& -\alpha_3 \|\bar{\hat{ x}}\|^2
		-2\bar{\hat{ x}}^T(1_N\otimes \dot{\bar x})
		-\varepsilon \|\tilde{x}\|^2  \\\nonumber
		&-\frac{\alpha_2}{\varepsilon^{n-1}} \|\bar y\|^2
		-\frac{\varepsilon}{\mu}\|\Delta\tilde{x}\|^2
		+\frac{1}{\varepsilon^{n-1}} {{\bar y}^T}{\bar f}\\	\nonumber
		&+\frac{k_1}{\varepsilon^{n-1}} {{\bar x}^T}\bigg({\bar f}-\alpha_1h-\alpha_2\bar y - \sum_{l=1}^{n-1}\frac{\varepsilon^nk_l}{\mu^{l+1-n}}\Delta\tilde{x}^{l}\bigg)	\\\nonumber
		& +\frac{1}{\varepsilon^{n-1}} \bigg(\sum_{l=1}^{n-2} (2 p_{l(n-1)} +k_{l+1}) \tilde{x}^{(l)}\\\nonumber
		& +(2 p_{(n-1)(n-1)} +2) \tilde{x}^{(n-1)}\bigg)^T (\bar f-\alpha_2 \bar y)  \\\nonumber
		&-\frac{1}{\varepsilon^{n-1}} \bigg(\sum_{l=1}^{n-2} (2 p_{l(n-1)} +k_{l+1}) \tilde{x}^{(l)}\nonumber \\
		&+(2 p_{(n-1)(n-1)} +1) \tilde{x}^{(n-1)}\bigg)^T(\alpha_1h+\sum_{l=1}^{n-1}\frac{\varepsilon^nk_l}{\mu^{l+1-n}}\Delta\tilde{x}^{l})\nonumber\\
		&-\frac{1}{\varepsilon^{n-1}} \bigg(2 r_{1n} \tilde{z}
		+\sum_{l=1}^{n-1} 2 r_{(l+1)n} \Delta\tilde{x}^{(l)}
		\bigg)^T (\bar f-\alpha_1 h-\alpha_2\bar y\nonumber \\& -\sum_{l=1}^{n-1}\frac{\varepsilon^nk_l}{\mu^{l+1-n}}\Delta\tilde{x}^{l}- \sum_{l=1}^{n-1}\varepsilon^nk_l\tilde{x}^{(l)}).
	\end{align*}

	


	
	By Assumption \ref{ass.non}, we have
	\begin{align}\label{eq2.rf}
		\nonumber
		&- \bigg(2 r_{1n} \tilde{z}
		+\sum_{l=1}^{n-1} 2 r_{(l+1)n} \Delta\tilde{x}^{(l)}
		\bigg)^T \bar f\\	\nonumber
		\leq& \frac{1}{2}\bigg(l_x+\sum_{l=1}^{n-1} l_{x^{(l)}}\varepsilon^{2l+1-n}\bigg)\bigg( 4 r_{1n}^2 \|\tilde{z}\|^2+\sum_{l=1}^{n-1} 4 r_{(l+1)n}^2 \|\Delta\tilde{x}^{(l)}\|^2\bigg) \\
		&+\frac{n}{2}\bigg(l_x\|\bar x\|^2+\sum_{l=1}^{n-1} l_{x^{(l)}}\varepsilon^{n-1} \|{\tilde x}^{(l)}\|^2\bigg).
	\end{align}
	
	Under Assumption \ref{ass.func}, we obtain
	\begin{align}\label{eq2.rh}
		\nonumber
		& \ \alpha_1\bigg(2 r_{1n} \tilde{z}
		+\sum_{l=1}^{n-1} 2 r_{(l+1)n} \Delta\tilde{x}^{(l)}
		\bigg)^T  h \\
		\leq& \frac{\theta\alpha_1}{2}\bigg(4 r_{1n}^2 \|\tilde{z}\|^2
		+\sum_{l=1}^{n-1} 4 r_{(l+1)n}^2 \|\Delta\tilde{x}^{(l)}\|^2 + n\|\bar{\hat x}\|^2\bigg).
	\end{align}
	
	Besides, we have
	\begin{subequations}\label{eq2.young}
		\begin{align}
	&-\frac{k_1\alpha_2}{\varepsilon^{n-1}}{\bar x}^T\bar y
			\leq
			\frac{{k_1}^2\alpha_2}{\varepsilon^{n-1}} \|\bar x\|^2
			+ \frac{\alpha_2}{4\varepsilon^{n-1}} \|\bar y\|^2 \label{eq2.xy} \\
	& -\sum_{l=1}^{n-1}\frac{\varepsilon k_1k_l}{\mu^{l+1-n}}{\bar x}^T\Delta\tilde{x}^{l}
			\leq \frac{nk_1\bar \mu }{2}\|\bar x\|^2 + \frac{k_1\varepsilon^2\bar \mu}{2}\|\Delta\tilde{x}\|^2 \label{eq2.xdetax} \\
			&-\alpha_2 \bigg(\sum_{l=1}^{n-2} (2 p_{l(n-1)} +k_{l+1}) {\tilde x}^{(l)} \nonumber \\
			& \qquad\quad   +(2 p_{(n-1)(n-1)} +2) {\tilde x}^{(n-1)} \bigg)^T\bar y \nonumber \\
			& \leq
			\alpha_2\bigg(\sum_{l=1}^{n-2} |2 p_{l(n-1)} +k_{l+1}| \|{\tilde x}^{(l)}\| \nonumber\\
			&\quad   +|2 p_{(n-1)(n-1)} +2| \|{\tilde x}^{(n-1)}\| \bigg)^T\|\bar y\| \nonumber \\
			&\leq
			\alpha_2\bigg({n}(\sum_{l=1}^{n-2} (2 p_{l(n-1)} +k_{l+1})^2 \|{\tilde x}^{(l)}\|^2\nonumber\\
			&\quad  \ +(2 p_{(n-1)(n-1)} +2)^2 \|{\tilde x}^{(n-1)}\|^2 )+ \frac{n-1}{4n}\|\bar y\|^2\bigg) \label{eq2.pxy} \\
					 & -\bigg(\sum_{l=1}^{n-2} (2 p_{l(n-1)} +k_{l+1}) {\tilde x}^{(l)} \nonumber \\
			& \qquad   +(2 p_{(n-1)(n-1)} +2) {\tilde x}^{(n-1)} \bigg)^T\sum_{l=1}^{n-1}\frac{\varepsilon k_l}{\mu^{l+1-n}}\Delta\tilde{x}^{(l)}\nonumber \\
			 & \leq  \frac{1}{16}\bigg(\sum_{l=1}^{n-2} (2 p_{l(n-1)} +k_{l+1})^2 \|{\tilde x}^{(l)}\|^2\nonumber\\
			&  \quad  +(2 p_{(n-1)(n-1)} +2)^2 \|{\tilde x}^{(n-1)}\|^2 \bigg)\nonumber\\
			&\quad + 4n(n-1)\varepsilon^2{\bar\mu}^2\|\Delta\tilde{x} \|^2 \label{eq2.pxdetax} \\
		&	\frac{\alpha_2}{\varepsilon^{n-1}} \bigg(2 r_{1n} \tilde{z}
			+\sum_{l=1}^{n-1} 2 r_{(l+1)n} \Delta\tilde{x}^{(l)}
			\bigg)^T  \bar y \nonumber\\
			 & \leq \frac{\alpha_2}{\varepsilon^{n-1}}\bigg(n(4 r_{1n}^2 \|\tilde{z}\|^2
			+\sum_{l=1}^{n-1} 4 r_{(l+1)n}^2 \|\Delta\tilde{x}^{(l)}\|^2 ) + \frac{1}{4}\|\bar y\|^2\bigg) \label{eq2.ry} \\
			& \varepsilon \bigg(2 r_{1n} \tilde{z}\nonumber
			+\sum_{l=1}^{n-1} 2 r_{(l+1)n} \Delta\tilde{x}^{(l)}
			\bigg)^T \sum_{l=1}^{n-1}k_l\tilde{x}^{(l)}\\
			&\leq \frac{\varepsilon}{2}\bigg(\sum_{l=1}^{n-1}2nk_l^2(4 r_{1n}^2 \|\tilde{z}\|^2
			+\sum_{l=1}^{n-1} 4 r_{(l+1)n}^2 \|\Delta\tilde{x}^{(l)}\|^2 ) +\|\tilde{x}\|^2 \bigg) \label{eq2.rx} \\
			& \bigg(2 r_{1n} \tilde{z}
			+\sum_{l=1}^{n-1} 2 r_{(l+1)n} \Delta\tilde{x}^{(l)}
			\bigg)^T \sum_{l=1}^{n-1}\frac{\varepsilon k_l}{\mu^{l+1-n}}\Delta\tilde{x}^{l} \nonumber\\
			&\leq \frac{n\varepsilon \bar \mu}{2}\bigg((4 r_{1n}^2+1) \|\tilde{z}\|^2
			+\sum_{l=1}^{n-1} (4 r_{(l+1)n}^2+1) \|\Delta\tilde{x}^{(l)}\|^2 \bigg).\label{eq2.rNx}
		\end{align}
	\end{subequations}
	
	It follows from \eqref{eq1.xh}, \eqref{eq1.ass2}, \eqref{eq1.hatxx},  \eqref{eq1.pxf}, \eqref{eq1.pxh} 
	and \eqref{eq2.rf}-\eqref{eq2.young} that
	\begin{align*}
		\dot{V}_2 \leq -\hat\rho_1\|\bar x\|^2 -\hat\rho_2\|{\tilde x}\|^2 -\hat\rho_3\|\bar y\|^2 -\hat\rho_4\|\bar{\hat{ x}}\|^2-\hat\rho_5\|\Delta\tilde{x}\|^2
	\end{align*}
	where
	\begin{align*}
		\hat\rho_1=&\frac{1}{2\varepsilon^{n-1}}(k_1\omega\alpha_1-2k_1^2\alpha_2-2l_x(n+k_1)\\
		&-\sum_{l=1}^{n-1}k_1l_{x^{(l)}}\varepsilon^{2l+1-n}-nk_1\bar\mu\varepsilon^{n-1} )\\
		\hat\rho_2=&\frac{1}{2\varepsilon^{n-1}}(\varepsilon^{n}-\varepsilon^{n-1}(\frac{\bar p}{16}+2N+(k_1+2n){l}_x^{max})\\
		&-\bar p(\theta\alpha_1 + 2n\alpha_2+l_x+\sum_{l=1}^{n-1}l_{x^{(l)}}\varepsilon^{2l+1-n}))\\
		\hat\rho_3=&\frac{1}{2\varepsilon^{n-1}}(\frac{n+1}{2n}\alpha_2 - l_x-\sum_{l=1}^{n-1}l_{x^{(l)}}\varepsilon^{2l+1-n})\\
		\hat\rho_4=&\alpha_3
		-\varepsilon^2 -\frac{\theta\alpha_1}{2\varepsilon^{n-1}}(2n-1+\frac{k_1\theta}{\omega})\\
		\hat\rho_5=&\frac{\varepsilon}{\mu}-\frac{\varepsilon^{2}\bar \mu}{2}(k_1+\frac{n\bar r}{\varepsilon}+8n(n-1)\bar\mu)\\
		&-\frac{\bar  r}{2\varepsilon^{n-1}}(\theta\alpha_1+2n\alpha_2+\sum_{l=1}^{n-1}2n\varepsilon^nk_l^2\\&+l_x+\sum_{l=1}^{n-1}l_{x^{(l)}}\varepsilon^{2l+1-n})
	\end{align*}
  $r_{in}$ is the last element of the $i$th row vector of matrix $R$; and
	$\bar r=  \max \big\{ 4r_{1n}^2+1, \ldots,4r_{nn}^2+1\big\}$.
	
	Similarly to the proof of Theorem \ref{lem.con1},
	we can take suitable $\alpha_1$, $\alpha_2$, $\alpha_3$ and $\varepsilon$ such that $\hat{\rho}_2$, $\hat{\rho}_3$, $\hat{\rho}_4>0$. Then, by selecting a small enough $\mu$, we have $\hat{\rho}_1$, $\hat{\rho}_5>0$.
	Therefore, there exist positive constants $\alpha_1$, $\alpha_2$,  $\alpha_3$, $\varepsilon$ and $\mu$ such that $\dot{V}_2\leq0$, which indicates that the high-order nonlinear player \eqref{sys} with the algorithm \eqref{o.al} exponentially converges to the Nash equilibrium of the noncooperative game \eqref{q}.
\end{proof}

\begin{remark}
Compared with the results in \cite{Frihauf2012NashES,Deng2019DistributedAF,Gadjov2019APA,Zhang2020DistributedNE}, the algorithms \eqref{s.al} and \eqref{o.al} are exponentially convergent, even in the presence of uncertain parameters.
Furthermore, the output-based algorithm \eqref{o.al} only relies on the output variables of high-order nonlinear players instead of all state variables.
\end{remark}

\begin{figurehere}
	\centering
	\begin{tikzpicture}[->,>=stealth]
		
		\tikzstyle{node1}=[shape=circle,
		fill=blue!20,
		draw=none,
		text=black,
		inner sep=2pt,
		minimum size = 16 pt,
		ball color = blue!40]
		\node[node1] (1) at (-2.8,1) {$1$};
		\node[node1] (2) at (-0.92,1) {$2$};
		\node[node1] (3) at (0.94,1) {$3$};
		\node[node1] (4) at (2.8,1) {$4$};
		\node[node1] (5) at (2.8,0) {$5$};
		\node[node1] (6) at (2.8,-1) {$6$};
		\node[node1] (7) at (0.94,-1) {$7$};
		\node[node1] (8) at (-0.92,-1) {$8$};
		\node[node1] (9) at (-2.8,-1) {$9$};
		\node[node1] (10) at (-2.8,0) {$10$};
		
		%
		\path (1) edge[->] (2)
		(2) edge[->] (3)
		(3) edge[->] (4)
		(4) edge[->] (5)
		(5) edge[->] (6)
		(6) edge[->] (7)
		(7) edge[->] (8)
		(8) edge[->] (9)
		(9) edge[->] (10)
		(10) edge[->] (1)
		(2) edge[->] (7)
		(8) edge[->] (3);
	\end{tikzpicture}
	\caption{The communication network of ten vehicles.
	}\label{fig.graphs}
	\hspace{0.1in}
\end{figurehere}

\section{Simulations}\label{sec.s}
In this section, simulation examples are given to illustrate the algorithms \eqref{s.al} and \eqref{o.al}.

\emph{Example 1:} Formation control of vehicles.

Consider the formation problem of ten vehicles over a weight-unbalanced digraph described as Fig. \ref{fig.graphs}.
The desired formation is constituted by the vehicles, if the relative positions between vehicles are the desired values.
That is, $lim_{t\rightarrow\infty}(p_i(t)-p_j(t))=d_{ij},\forall i,j\in\mathcal{V}$, where $p_i$, $p_j$ are the position of vehicles $i$ and $j$, respectively; $d_{ij}\in\mathbb{R}^m$ is the desired relative position between vehicles $i$ and $j$.
According to  \cite{Deng2021GF}, the vehicles can form the desired formation by seeking the Nash equilibrium of the following game.
\begin{align}\label{g.for}
	\min \limits_{{p_i}\in\mathbb{R}^{2}}\mathop{J_i(p_i,p_{-i})}
\end{align}
where $J_i(p_i,p_{-i}) := \frac{1}{2N}\|p_i-2d_i\|^2 + \frac{1}{N}p_i^T\sum_{j=1}^Np_j$ and $p_i:=[p_{xi}\ p_{yi}]^T$ are the cost function and the position of vehicle $i$, respectively; $d_i := [d_{xi}\ d_{yi}]^T\in\mathbb{R}^2$ and $d_j := [d_{xj}\  d_{yj}]^T\in\mathbb{R}^2$ are constant vectors such that $d_i-d_j=d_{ij}$ holds.

The vehicle $i$ has following dynamics  (see \cite{2006Vehicle}).
\begin{align*}
	\dot{p}_i &= v_i \\
	\dot{v}_i &=  - \frac{\rho A_iC_{di}}{2m_i}v_i^2 - \frac{d_{mi}}{m_i} +u_i
\end{align*}
where $v_i$ is the speed of vehicle $i$, $u_i$ is the force produced by the engine; $m_i$, $A_i$ and $C_{di}$ are the mass, the cross-sectional area, and the drag coefficient of vehicle $i$, respectively; $\rho$ is air density; and $d_{mi}$ is a constant representing the amplitude of the mechanical drag force.

The parameters of ten vehicles are given in Table \ref{tab.veh}, and $\rho=1.225m/s^3$. The desired formation is a five-pointed star.

Fig. \ref{fig.form_state_s}  shows the simulation results of algorithm \eqref{s.al} with parameters $\alpha_1=3$, $\alpha_2=2.2$, $\alpha_3=18$ and $\varepsilon=2$.
Fig. \ref{fig.form_output_s} displays the simulation results of algorithm \eqref{o.al} with parameters $\alpha_1=3$, $\alpha_2=2.2$, $\alpha_3=18$, $\varepsilon=2$ and $\mu=0.02$.
In Figs. \ref{fig.form_state_s} and \ref{fig.form_output_s}, the triangles and the dots are the initial and final positions of the vehicles, respectively, and the solid lines are the position trajectories of vehicles.
As shown in Figs. \ref{fig.form_state_s} and \ref{fig.form_output_s}, the ten vehicles converge to a five-pointed star under the algorithms \eqref{s.al} and \eqref{o.al}, respectively.
By comparing Fig. \ref{fig.form_state_s} with Fig. \ref{fig.form_output_s}, the ten vehicles under the algorithm \eqref{s.al} form the desired formation more directly than that under the algorithm \eqref{o.al}.

\begin{tablehere}
	\setlength{\abovecaptionskip}{-0cm}   
	\setlength{\belowcaptionskip}{0cm}
	\begin{center}
	\caption{The parameters of ten vehicles }\label{tab.veh}
	\begin{tabular}{cccccccccc}
	\hline
	Vehicle $i$& $m_i(kg)$ & $A_i(m^2)$ & $C_{di}$ & $d_{mi}(N)$&   \\
    \hline
	 $1$&   $1800$  &  $2.180$& $1.526$&$6.412$& \\
	 $2$&  $1775$  &  $2.165$& $1.649$&$5.241$&   \\
	 $3$& $825$  &  $1.634$& $1.052$&$2.466$&   \\
	 $4$&  $1025$  &  $1.746$& $1.281$&$3.969$&  \\
	 $5$&  $1200$   &  $1.844$& $1.359$&$4.113$&\\
	 $6$&  $1450$   &  $1.983$& $1.420$&$4.755$&  \\
	 $7$&  $970$   &  $1.715$& $1.138$&$2.842$& \\
	 $8$&  $1500$   &  $2.011$& $1.409$&$4.672$& \\
	 $9$& $1320$   &  $1.911$& $1.389$&$4.263$&  \\
	 $10$&  $1670$  &  $2.107$& $1.514$&$5.038$&   \\
	\hline
	\end{tabular}
	\end{center}
\end{tablehere}

\begin{figurehere}
	\centering
\setlength{\abovecaptionskip}{-0.2cm}
 \setlength{\belowcaptionskip}{0cm}
	\includegraphics[width=7.5cm]{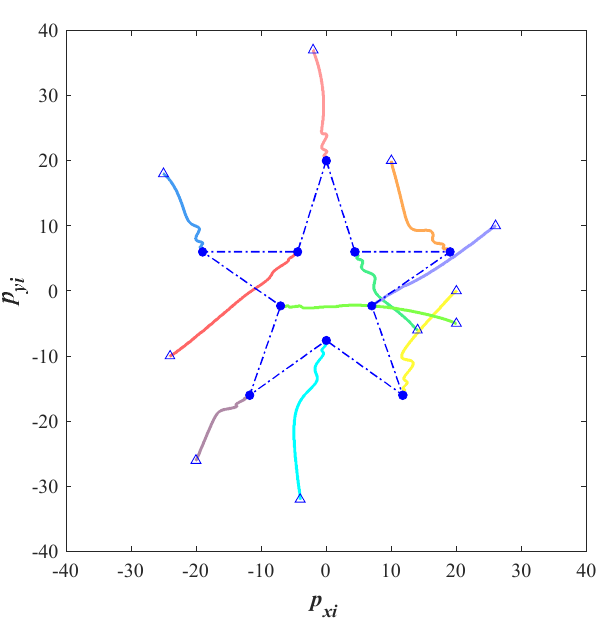}
	\caption{The  position trajectories of ten vehicles under algorithm (\ref{s.al}).}
	\label{fig.form_state_s}
\end{figurehere}

\begin{figurehere}
	\centering
\setlength{\abovecaptionskip}{-0.2cm}
 \setlength{\belowcaptionskip}{0cm}
	\includegraphics[width=7.5cm]{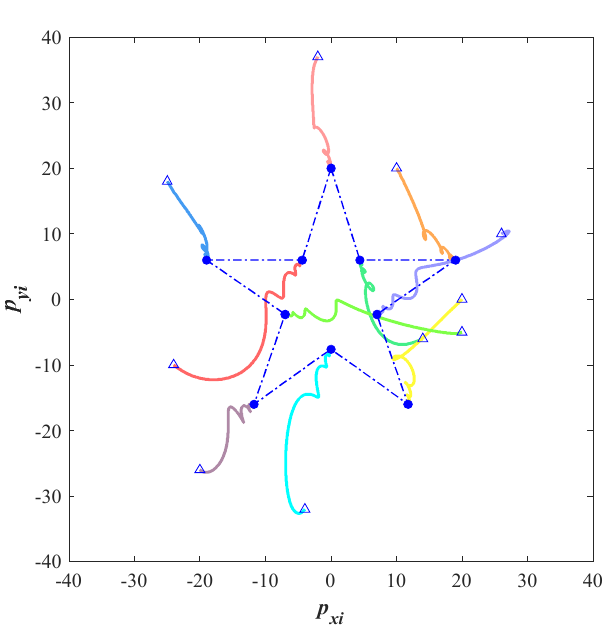}
	\caption{The position  trajectories of ten vehicles under algorithm (\ref{o.al}).}
	\label{fig.form_output_s}
\end{figurehere}

\begin{figurehere}
	\centering
	\begin{tikzpicture}[->,>=stealth]

		\tikzstyle{node1}=[shape=circle,
		fill=green!20,
		draw=none,
		text=black,
		inner sep=2pt,
		minimum size = 16 pt,
		ball color = green!40]
		\tikzstyle{node3}=[shape=circle,
		fill=red!20,
		draw=none,
		text=black,
		inner sep=2pt,
		minimum size = 16 pt,
		ball color = red!40]

		\node[node1] (1) at (-1,1) {$1$};
		\node[node1] (2) at (1,1) {$2$};
		\node[node1] (3) at (2,0) {$3$};
		\node[node1] (4) at (1,-1) {$4$};
		\node[node1] (5) at (-1,-1) {$5$};
		\node[node1] (6) at (-2,0) {$6$};

		\path (1) edge[->] (2)
		(2) edge[->] (3)
		(3) edge[->] (4)
		(4) edge[->] (5)
		(5) edge[->] (6)
		(6) edge[->] (1)
		(1) edge[->] (5)
		(4) edge[->] (2);
	\end{tikzpicture}
	\caption{The communication network of six turbine-generator systems.
	}\label{fig.graphs2}
	\hspace{0.1in}
\end{figurehere}

\begin{tablehere}
 \setlength{\abovecaptionskip}{-0cm}   
 \setlength{\belowcaptionskip}{0cm}
 \begin{center}
  \caption{The parameters of six turbine-generator systems}\label{tab.tur}
  \begin{tabular}{cccccccccc}
   \hline
   Generator  $i$& $\gamma_{i1}$& $\gamma_{i2}$ & $\gamma_{i3}$ &\\
   \hline
   $1$&   $7$  &  $36.80$& $0.27$& \\
   $2$&  $20$  &  $13.73$& $0.15$&   \\
   $3$& $60$  &  $17.14$& $0.23$& \\
   $4$&  $15$  &  $20.41$& $0.10$&  \\
   $5$&  $10$   &  $15.28$& $0.18$&\\
   $6$&  $55$   &  $14.07$& $0.32$&  \\
   \hline
  \end{tabular}
 \end{center}
\end{tablehere}

\emph{Example 2:} Electricity market games of turbine-generator systems.

In electricity markets, the competition among distributed energy resources can be described by noncooperative games (see \cite{Liu2018OptimalDC}).
Consider a noncooperative game of six turbine-generator systems over a weight-unbalanced digraph depicted as Fig. \ref{fig.graphs2}.
The turbine-generator system $i\in\mathcal{V}$ faces the following game (see \cite{Deng2019DistributedAF}).
\begin{equation*}
	\min \limits_{{P_i}\in\mathbb{R}}{J_i(P_i,P_{-i})}
\end{equation*}
where $J_i:\mathbb{R}\rightarrow\mathbb{R}$ and $P_i$ are the cost function and the output power of the turbine-generator system $i$, respectively; $P_{-i}=col(P_1,\dots,P_{i-1},P_{i+1},\dots,P_{N})$.
Specifically, the cost function of turbine-generator system $i$ is
\begin{equation*}
	J_i(P_i,P_{-i})=c_i(P_i)-p(\sigma)P_i
\end{equation*}
where $c_i(P_i):=\gamma_{i1}+\gamma_{i2}P_i+\gamma_{i3}P_i^2$ is the generation cost with $\gamma_{i1}$, $\gamma_{i2}$, $\gamma_{i3}$ being the characteristics of the generation system $i$ presented in Table \ref{tab.tur} (see \cite{Binetti2014DistributedCE}); $p(\sigma):=200- 0.1N \sigma$ is the electricity price with $\sigma(P):=\frac{1}{N}\sum_{i=1}^NP_i$.

When the valve positions of governors are fixed, the turbine-generator system $i$ have fourth-order dynamics $P_i^{(4)} = u_i$ (see \cite{Deng2021DAD,Fang2011BacksteppingbasedNA}).

\begin{figurehere}
	\centering
\setlength{\belowcaptionskip}{0.25cm}
	\includegraphics[width=9cm]{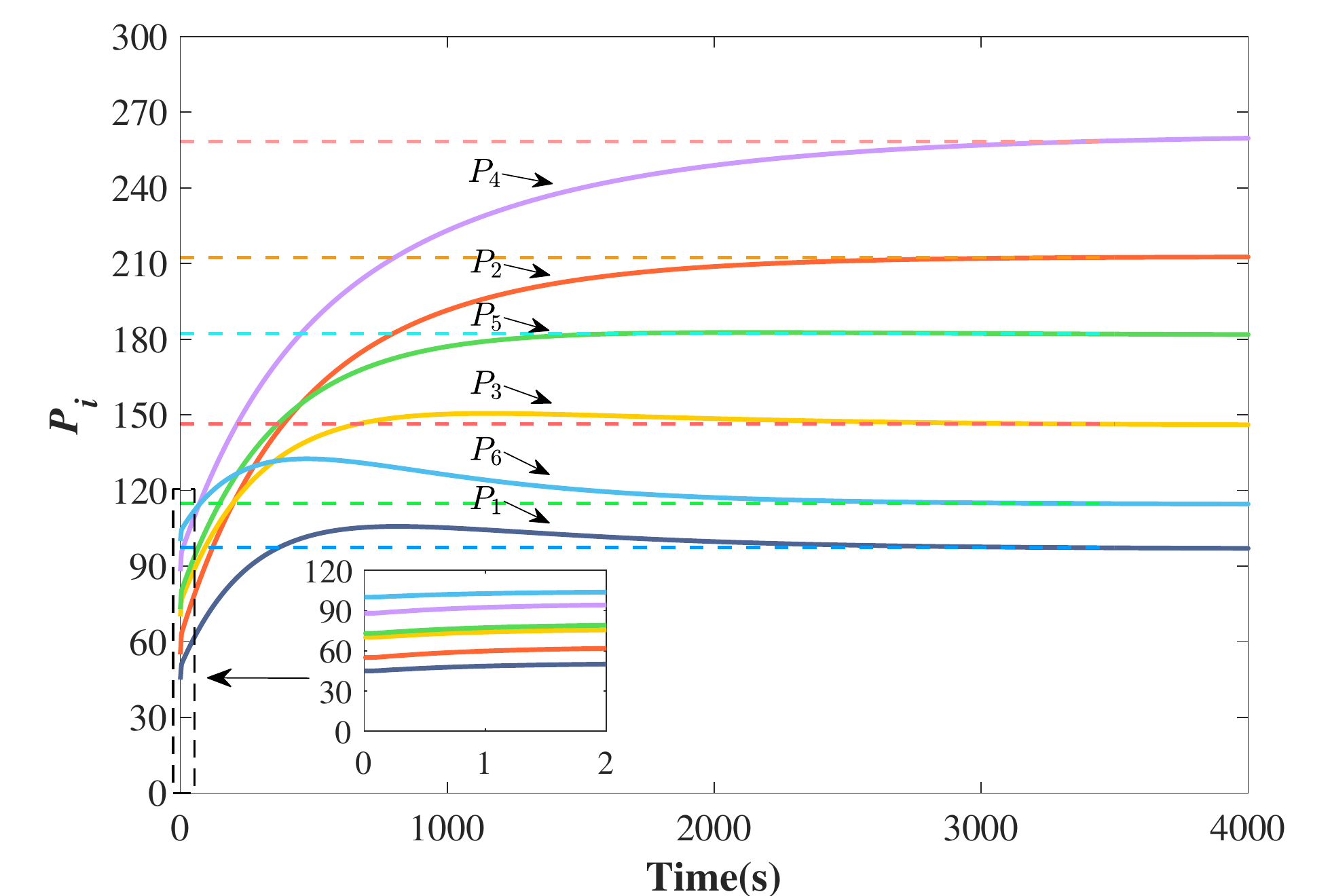}
	\caption{The  evolutions of output powers under algorithm (\ref{s.al}).}
	\label{fig.generator_state_s}
\end{figurehere}

\begin{figurehere}
	\centering
\setlength{\belowcaptionskip}{0.25cm}
	\includegraphics[width=9cm]{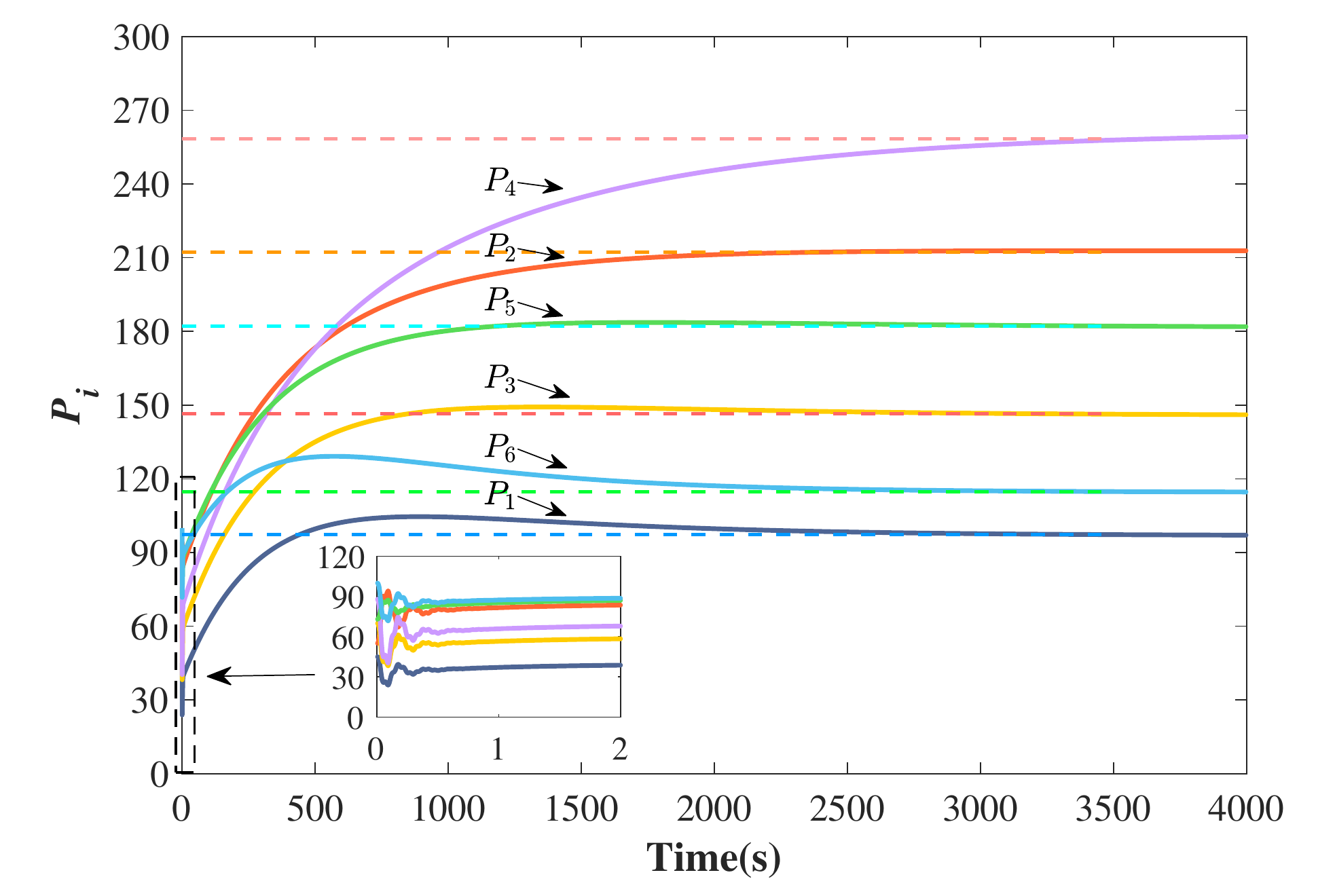}
	\caption{The  evolutions of output powers under algorithm (\ref{o.al}).}
	\label{fig.generator_output_s}
\end{figurehere}

Fig. \ref{fig.generator_state_s}  displays the simulation results of algorithm \eqref{s.al} with parameters $\alpha_1=500$, $\alpha_2=400$, $\alpha_3=400$ and $\varepsilon=20$.
Fig. \ref{fig.generator_output_s} shows the simulation results of algorithm \eqref{o.al} with parameters $\alpha_1=500$, $\alpha_2=400$, $\alpha_3=400$, $\varepsilon=20$ and $\mu=0.01$.
In  Figs. \ref{fig.generator_state_s} and \ref{fig.generator_output_s}, the solid lines and the dotted lines are the evolutions of output powers and the Nash equilibrium, respectively.
It is obvious from  Figs. \ref{fig.generator_state_s} and \ref{fig.generator_output_s} that the output powers of six turbine-generator systems converge to the Nash equilibrium under the algorithms \eqref{s.al} and \eqref{o.al}.
By comparing Fig. \ref{fig.generator_state_s} with Fig. \ref{fig.generator_output_s}, the algorithm \eqref{s.al} has better performance than the algorithm \eqref{o.al}, since the fluctuation of simulation results under algorithm \eqref{o.al} is more drastic than that under algorithm \eqref{s.al}, especially in initial stage.
These simulation results verify the effectiveness of the algorithms \eqref{s.al} and \eqref{o.al}.

\section{Conclusions}\label{sec.c}
This paper has studied the noncooperative games of multi-agent systems over weight-unbalanced digraphs, where the players have high-order uncertain nonlinear dynamics.
In order to guarantee that the high-order nonlinear players can autonomously seek the Nash equilibrium of the game, this paper has developed a distributed state-based algorithm and a distributed output-based algorithm.
The state-based algorithm uses all state variables of high-order nonlinear players, and the output-based algorithm only needs the output variables.
Moreover, the paper has analyzed the globally exponentially convergence of the two algorithms.
Finally, the two algorithms have been verified by  two examples about the formation problem of vehicles and the electricity market games of smart grids.

\footnotesize
\section*{References}

\bibliographystyle{elsarticle-num}
\bibliography{reference}

\begin{thebibliography}{10}
\expandafter\ifx\csname url\endcsname\relax
  \def\url#1{\texttt{#1}}\fi
\expandafter\ifx\csname urlprefix\endcsname\relax\def\urlprefix{URL }\fi
\expandafter\ifx\csname href\endcsname\relax
  \def\href#1#2{#2} \def\path#1{#1}\fi

\bibitem{Jensen2010AggregativeGA}
M.~K. Jensen, Aggregative games and best-reply potentials, Economic Theory
  43~(1) (2010) 45--66.

\bibitem{Gharesifard2016PriceBasedCA}
B.~Gharesifard, T.~Başar, A.~D. Dom{\'i}nguez-Garc{\'i}a, Price-based
  coordinated aggregation of networked distributed energy resources, IEEE
  Transactions on Automatic Control 61~(10) (2016) 2936--2946.

\bibitem{Deng2021DistributedSM}
Z.~Deng, Distributed algorithm design for aggregative games of {Euler-Lagrange}
  systems and its application to smart grids, IEEE transactions on Cybernetics
  8~(1) (2021) 177--186.

\bibitem{Barrera2015DynamicIF}
J.~Barrera, A.~Garc{\'i}a, Dynamic incentives for congestion control, IEEE
  Transactions on Automatic Control 60~(2) (2015) 299--310.

\bibitem{Lou2016NashEC}
Y.~Lou, Y.~Hong, L.~Xie, G.~Shi, K.~H. Johansson, Nash equilibrium computation
  in subnetwork zero-sum games with switching communications, IEEE Transactions
  on Automatic Control 61~(10) (2016) 2920--2935.

\bibitem{Ghaderi2014OpinionDI}
J.~Ghaderi, R.~Srikant, Opinion dynamics in social networks with stubborn
  agents: Equilibrium and convergence rate, Automatica 50~(12) (2014)
  3209--3215.

\bibitem{Frihauf2012NashES}
P.~Frihauf, M.~Krsti{\'c}, T.~Başar, Nash equilibrium seeking in
  noncooperative games, IEEE Transactions on Automatic Control 57~(5) (2012)
  1192--1207.

\bibitem{Facchinei2010GeneralizedNE}
F.~Facchinei, C.~Kanzow, Generalized {Nash} equilibrium problems, Annals of
  Operations Research 175 (2010) 177--211.

\bibitem{Li2020DistributedGN}
Z.~Li, Z.~Li, Z.~Ding, Distributed generalized {Nash} equilibrium seeking and
  its application to femtocell networks, IEEE transactions on Cybernetics PP
  (2020) 1--13.

\bibitem{Koshal2016DistributedAF}
J.~Koshal, A.~Nedi{\'c}, U.~V. Shanbhag, Distributed algorithms for aggregative
  games on graphs, Operations Research 64 (2016) 680--704.

\bibitem{Deng2019DistributedAF}
Z.~Deng, S.~Liang, Distributed algorithms for aggregative games of multiple
  heterogeneous {Euler-Lagrange} systems, Automatica 99 (2019) 246--252.

\bibitem{ZENG2019GNE}
X.~Zeng, J.~Chen, S.~Liang, Y.~Hong, Generalized {Nash} equilibrium seeking
  strategy for distributed nonsmooth multi-cluster game, Automatica 103 (2019)
  20--26.

\bibitem{Gadjov2019APA}
D.~Gadjov, L.~Pavel, A passivity-based approach to {Nash} equilibrium seeking
  over networks, IEEE Transactions on Automatic Control 64~(3) (2019)
  1077--1092.

\bibitem{Lu2019DistributedAF}
K.~Lu, G.~Jing, L.~Wang, Distributed algorithms for searching generalized
  {Nash} equilibrium of noncooperative games, IEEE Transactions on Cybernetics
  49~(6) (2019) 2362--2371.

\bibitem{Ye2017DistributedNE}
M.~Ye, G.~Hu, Distributed {Nash} equilibrium seeking by a consensus based
  approach, IEEE Transactions on Automatic Control 62~(9) (2017) 4811--4818.

\bibitem{Deng2021nonsmooth}
Z.~Deng, Distributed generalized {Nash} equilibrium seeking algorithm for
  nonsmooth aggregative games, Automatica 132 (2021) 109794.

\bibitem{Zhang2020DistributedNE}
Y.~Zhang, S.~Liang, X.~Wang, H.~Ji, Distributed {Nash} equilibrium seeking for
  aggregative games with nonlinear dynamics under external disturbances, IEEE
  Transactions on Cybernetics 50~(12) (2020) 4876--4885.

\bibitem{Zhu2019ContinuousTimeCA}
Y.~Zhu, W.~Yu, G.~Wen, W.~Ren, Continuous-time coordination algorithm for
  distributed convex optimization over weight-unbalanced directed networks,
  IEEE Transactions on Circuits and Systems II: Express Briefs 66~(7) (2019)
  1202--1206.

\bibitem{Kim2012CyberPhysicalSA}
K.-D. Kim, P.~R. Kumar, Cyber–physical systems: A perspective at the
  centennial, Proceedings of the IEEE 100 (2012) 1287--1308.

\bibitem{Zhang2018DCR}
X.~Zhang, A.~Papachristodoulou, N.~Li, Distributed control for reaching optimal
  steady state in network systems: An optimization approach, IEEE Transactions
  on Automatic Control 63~(3) (2018) 864--871.

\bibitem{Zhang2015DistributedPA}
Y.~Zhang, Y.~Lou, Y.~Hong, L.~Xie, Distributed projection-based algorithms for
  source localization in wireless sensor networks, IEEE Transactions on
  Wireless Communications 14~(6) (2015) 3131--3142.

\bibitem{Wang2021DOC}
Q.~Wang, J.~Chen, B.~Xin, X.~Zeng, Distributed optimal consensus for
  {Euler–Lagrange} systems based on event-triggered control, IEEE
  Transactions on Systems, Man, and Cybernetics: Systems 51~(7) (2021)
  4588--4598.

\bibitem{Zhang2017DistributedOC}
Y.~Zhang, Z.~Deng, Y.~Hong, Distributed optimal coordination for multiple
  heterogeneous {Euler-Lagrangian} systems, Automatica 79 (2017) 207--213.

\bibitem{Wang2020DistributedPA}
Q.~Wang, J.~Chen, X.~Zeng, B.~Xin, Distributed proximal‐gradient algorithms
  for nonsmooth convex optimization of second‐order multiagent systems,
  International Journal of Robust and Nonlinear Control 30 (2020) 7574--7592.

\bibitem{Deng2021DAD}
Z.~Deng, Distributed algorithm design for resource allocation problems of
  high-order multiagent systems, IEEE Transactions on Control of Network
  Systems 8~(1) (2021) 177--186.

\bibitem{Deng2021DistributedEA}
Z.~Deng, Distributed {Nash} equilibrium seeking for aggregative games with
  second-order nonlinear players, Automatica 135 (2022) 109980.

\bibitem{Romano2020DynamicNS}
A.~R. Romano, L.~Pavel, Dynamic {NE} seeking for multi-integrator networked
  agents with disturbance rejection, IEEE Transactions on Control of Network
  Systems 7~(1) (2020) 129--139.

\bibitem{Eichhorn1998TransformationsON}
R.~Eichhorn, S.~J. Linz, P.~Hnggi, Transformations of nonlinear dynamical
  systems to jerky motion and its application to minimal chaotic flows,
  Physical review. E, Statistical physics, plasmas, fluids, and related
  interdisciplinary topics 58~(6) (1998) 7151--7164.

\bibitem{2002Nonlinear}
H.~K. Khalil, Nonlinear Systems, 3rd Edition, Prentice Hall, Upper Saddle
  River, N.J., 2002.

\bibitem{2006Vehicle}
R.~Rajamani, Vehicle Dynamics and Control, Springer, Boston, MA, New York,
  2006.

\bibitem{Fang2011BacksteppingbasedNA}
F.~Fang, L.~Wei, Backstepping-based nonlinear adaptive control for coal-fired
  utility boiler–turbine units, Applied Energy 88~(3) (2011) 814--824.

\bibitem{Li2017OFD}
Y.~Li, C.~Hua, S.~Wu, X.~Guan, Output feedback distributed containment control
  for high-order nonlinear multiagent systems, IEEE Transactions on Cybernetics
  47~(8) (2017) 2032--2043.

\bibitem{Rezaee2019StationaryCC}
H.~Rezaee, F.~Abdollahi, Stationary consensus control of a class of high-order
  uncertain nonlinear agents with communication delays, IEEE Transactions on
  Systems, Man, and Cybernetics: Systems 49~(6) (2019) 1285--1290.

\bibitem{GODSIL2004Algebraic}
C.~D. Godsil, G.~Royle, Algebraic Graph Theory, Springer, New York, 2001.

\bibitem{Hu2007LeaderfollowingCO}
J.~Hu, Y.~Hong, Leader-following coordination of multi-agent systems with
  coupling time delays, Physica A-statistical Mechanics and Its Applications
  374~(2) (2007) 853--863.

\bibitem{2004Variational}
R.~T. Rockafellar, R.~J.~B. Wets, Variational Analysis, 2nd Edition, Springer,
  Berlin, 2004.

\bibitem{2003Finite}
F.~Facchinei, J.-S. Pang, Finite-Dimensional Variational Inequalities and
  Complementarity Problems, Springer, New York, 2003.

\bibitem{Ghapani19}
S.~Ghapani, S.~Rahili, W.~Ren, Distributed average tracking of physical
  second-order agents with heterogeneous unknown nonlinear dynamics without
  constraint on input signals, IEEE Transactions on Automatic Control 64~(3)
  (2019) 1178--1184.

\bibitem{Yin2011NashEP}
H.~Yin, U.~V. Shanbhag, P.~G. Mehta, {N}ash equilibrium problems with scaled
  congestion costs and shared constraints, IEEE Transactions on Automatic
  Control 56~(7) (2011) 1702--1708.

\bibitem{2010Generalized}
F.~Facchinei, C.~Kanzow, Generalized {N}ash equilibrium problems, Annals of
  Operations Research 175~(1) (2010) 177--211.

\bibitem{1998Linear}
C.-T. Chen, Linear System Theory and Design, 3rd Edition, Oxford University
  Press, New York, 1999.

\bibitem{Deng2021GF}
Z.~Deng, Game-based formation control of high-order multi-agent systems,
  submitted to IEEE Transaction on Cybernetics.

\bibitem{Liu2018OptimalDC}
Z.~Liu, Q.~Wu, S.~Huang, L.~Wang, M.~Shahidehpour, Y.~Xue, Optimal day-ahead
  charging scheduling of electric vehicles through an aggregative game model,
  IEEE Transactions on Smart Grid 9~(5) (2018) 5173--5184.

\bibitem{Binetti2014DistributedCE}
G.~Binetti, A.~Davoudi, F.~L. Lewis, D.~Naso, B.~Turchiano, Distributed
  consensus-based economic dispatch with transmission losses, IEEE Transactions
  on Power Systems 29~(4) (2014) 1711--1720.

\end{thebibliography}

%
%

%
\end{multicols}







\end{document}